\documentclass[onecolumn,12pt,draftclsnofoot]{IEEEtran}
\usepackage{amssymb,amsmath,amsfonts,amsthm}
\usepackage{graphicx}
\usepackage[extra]{tipa}
\usepackage{url}
\usepackage{subfigure}
\usepackage{setspace}
\usepackage{color}
\newtheorem{thm}{Theorem}[section]
\newtheorem{corollary}[thm]{Corollary}
\newtheorem{lemma}[thm]{Lemma}

\newtheorem{remark}{Remark}
\IEEEoverridecommandlockouts

\ifCLASSINFOpdf
\else
\fi

\begin{document}

\title{Stable Throughput and Delay Analysis of a Random Access Network With Queue-Aware Transmission}
\author{\IEEEauthorblockN{Ioannis Dimitriou, Nikolaos Pappas~\IEEEmembership{Member,~IEEE}
\thanks{I. Dimitriou is with the Department of Mathematics, University of Patras, Patra, Peloponnese, Greece.
(e-mail: idimit@math.upatras.gr).}
\thanks{N. Pappas is with the Department of Science and Technology, Link\"{o}ping University, Norrk\"{o}ping SE-60174, Sweden. (e-mail: nikolaos.pappas@liu.se).}
}}

\maketitle

\begin{abstract}
In this work we consider a two-user and a three-user slotted ALOHA network with multi-packet reception (MPR) capabilities. The nodes can adapt their transmission probabilities and their transmission parameters based on the status of the other nodes. Each user has external bursty arrivals that are stored in their infinite capacity queues. For the two- and the three-user cases we obtain the stability region of the system. For the two-user case we provide the conditions where the stability region is a convex set. We perform a detailed mathematical analysis in order to study the queueing delay by formulating two boundary value problems (a Dirichlet and a Riemann-Hilbert boundary value problem), the solution of which provides the generating function of the joint stationary probability distribution of the queue size at user nodes. Furthermore, for the two-user symmetric case with MPR we obtain a lower and an upper bound for the average delay without explicitly computing the generating function for the stationary joint queue length distribution. The bounds as it is seen in the numerical results appear to be tight. Explicit expressions for the average delay are obtained for the symmetrical model with capture effect which is a subclass of MPR models. We also provide the optimal transmission probability in closed form expression that minimizes the average delay in the symmetric capture case. Finally, we evaluate numerically the presented theoretical results.
\end{abstract}


\section{Introduction}

The ALOHA protocol since its creation \cite{Abramson} has gained popularity in multiple access communication systems for its simple nature and the fact that it does not require centralized controllers. This simple scheme attempts transmission randomly, independently, distributively, and based on a simple ACK/NACK feedback from the receiver.

Random access recently re-gained interest due to the increase in the number of communicating devices in 5G networks, more specifically, because of the need of massive uncoordinated access in large networks \cite{Laya2014, METISComMag2014}. Random access and alternatives and their effect on the operation of LTE and LTE-A are presented in \cite{Laya2014,KoseogluTCOM2016,LeungTWC2012}. Recently, the effect of random access in physical layer and in other topics has been studied \cite{PopovskiSPL2017,WangTSP2007,PopovskiWCL2015,PopovskiTCOM2013} and the research in this area is in progress.
Random access remains an active research area where a lot of fundamental questions remain open even for very simple networks \cite{AEUnion,TongSingProcMag2004}.

When the traffic in a network is bursty, a relevant performance measure is the stable throughput or stability region. The exact characterization of the stability region is known to be a difficult problem due to the interaction among the queues. Except the throughput, delay is another important metric. Recently there is a rapid growth on supporting real-time applications thus, there is a need to provide delay-based guarantees \cite{METISComMag2014, KumarTMC2016}. Thus, the characterization of the delay is of major importance. However, the exact characterization of delay even in small networks with random access is rather difficult and remains unexplored in most of the cases.
 
In this work, we consider a two-user and a three-user slotted ALOHA network with multi-packet reception capabilities. Furthermore, the nodes can adapt their transmission probabilities and their transmission parameters based on the status of the other node. We analyze the stable throughput region and study the queueing delay by utilizing the theory of boundary value problems. 
 
\subsection{Related Work}

In the literature so far there is a vast number of papers that are considering the stable throughput and delay in random access and variations of random access schemes.

The derivation of the stability region of random access systems for bursty sources is known to be a difficult problem above three sources. This is because each source transmits and interferes with the others only when its queue is non-empty. Such queues where the service process of one depends on the status of the others are said to be coupled or interacting. Thus, the individual departure rates of the queues cannot be computed separately without knowing the stationary distribution of the joint queue length process \cite{Rao_TIT1988}. This is the reason why the vast majority of previous works has focused on small-sized networks and only bounds or approximations are known for the networks with larger number of sources \cite{Tsybakov79, Rao_TIT1988, szpa1, LuoAE1999, NawareTong2005, WangTSP2007}. In \cite{Bordenave_TIT2012}, an approximation of the stability region was obtained based on the mean-field theory for network of nodes having identical arrival rates and transmission probabilities were performed. The work in \cite{Haenggi_TCOM2016} investigates the stable throughput region of a random access network where the transmitters and receivers are distributed by a static Poisson bipolar process.

Delay analysis of random access networks was studied in \cite{NawareTong2005, BehrooziRaoTIT1992, GeorgiadisJSAC87}. More specifically, in \cite{NawareTong2005,WangTSP2007} a two-user network with MPR capabilities was considered and expressions for the average delay were obtained for the symmetric case. The papers \cite{BehrooziRaoTIT1992, GeorgiadisJSAC87} considered collision channel model. In \cite{StamatiouHaenggi2010} the delay performance of slotted ALOHA in a Poisson network was studied. Delay analysis of random access networks based on fluid models can be found in \cite{WangTong2010} and in \cite{ProutiereITC2011}. The works \cite{HaenggiAllerton2010} and \cite{HaenggiTMC2012} utilized techniques from statistical mechanics for throughput and delay analysis. The authors in \cite{PoloczekInfocom2015} proposed a service-martingale concept that enables the queueing analysis of a bursty source sharing a MAC channel.

Below we present a recent set of papers that consider throughput and/or delay characterization of general random access networks. The work in \cite{PappasTWC2015} studied the impact of a full duplex relay in terms of throughput and delay in a multi-user network, where the users were assumed to have saturated traffic. The delay of a random access scheme in the Internet of Things concept was studied in \cite{TanJIOT2017}. In \cite{ChenWoWMoM2016} throughput with delay constraints was studied in a shared access cognitive network.
The delay characterization of larger networks was considered in \cite{StamatiouTON2014,QuekTWC2016}. In \cite{SandgrenTCOM2017} the delay and the packet loss rate of a frame asynchronous coded slotted ALOHA system for an uncoordinated multiple access were also studied.

\subsection{Contribution}
Our contribution in this work can be summarized as follows. We consider the case of the two and three-user wireless network with a common destination. The nodes/sources access the medium in a random access manner and time is assumed to be slotted. Each user has external bursty arrivals that are stored in their infinite capacity queues. We consider multi packet reception (MPR) capabilities at the destination node.

The nodes are accessing the wireless channel randomly and they adapt their transmission probabilities based on the status of the queue of the other nodes. More precisely, a node adapt its transmission characteristics based on the status of the other node in order to exploits its idle slots and to increase the chances of a successful packet transmission. To the best of our knowledge this variation of random access has not been reported in the literature. The contribution of this work has two main parts focused on the stable throughput region and the detailed analysis of the queueing delay at users nodes.

\subsubsection{Stable Throughput Region Analysis}

The first part is related to the study of stable throughput.
\begin{itemize}
\item More specifically, we obtain the stability conditions for the case of two and three users. To the best of our knowledge, there is no other work in the related literature that deals with the stability region of a random access system with adaptive transmission control. Furthermore, we obtain the conditions where the stability region is a convex set. Convexity is an important property since it corresponds to the case when parallel concurrent transmissions are preferable to a time-sharing scheme.
\item We would like also to emphasize that the exact stability region for the case of three nodes with MPR even in the simple random access (without transmission control) case is not known in the literature, except for the case of a collision channel model \cite{szpa}.
\end{itemize}
  
The main difficulty for characterizing the stability region lies on the interaction of the queues. The interaction of the queues arise when the service rate of a queue depends on the state of the other. A tool to bypass this difficulty is the stochastic dominance technique introduced in \cite{Rao_TIT1988}.
However, the three-user network is more elaborated and the stability region cannot be derived that easily. As mentioned also earlier, in the literature the three-user scenario has studied only for the collision channel model.

\subsubsection{Delay Analysis}
The second part of the contribution of this work is the delay analysis.
\begin{itemize}
\item Based on a relation among the values of the transmission probabilities we distinguish the analysis in two cases, which are different in the modeling and the technical point of view. In particular, the analysis leads to the formulation of two boundary value problems (e.g., \cite{coh,fay,nain,fay1,avr}), the solution of which will provide the joint probability distribution of the queue size for the two-user case with MPR. The analysis is rather complicated and novel.
\item Furthermore, for the two-user symmetric case with MPR we obtain a lower and an upper bound for the average delay without explicitly computing the generating function for the stationary joint queue length distribution.
\item The bounds as it is seen in the numerical results appear to be tight. Explicit expressions for the average delay are obtained for the model with capture effect, i.e., a subclass of MPR models.
\item We also provide the optimal transmission probability in closed form expression that minimizes the average delay in the symmetric capture case.
\end{itemize}

Concluding, the analytical results in this work, to the best of our knowledge, have not been reported in the literature.

The rest of the paper is organized as follows. In Section \ref{sec:model} we present the system model by providing the details of the proposed protocol and the underlying physical layer details on the channel model. In Section \ref{sec:stability2} we provide the stability region for the two-user case for the proposed random access scheme. In Section \ref{sec:analysis} we derive the fundamental functional equation and obtain some important results for the following analysis. Section \ref{sec:bound} is devoted to the formulation of two boundary value problems, the solution of which provides the generating function of the joint stationary queue length distribution of user nodes. The expected number of packets and the average delay expressions are also obtained. In Section \ref{sec:stabb}, we provide an alternative approach to obtain the stability conditions and we also obtain the exact expressions for the stability region for the case of three users. In Section \ref{sec:symmetric2user}, we obtain explicit expressions for the average delay at each user for the symmetrical system. Finally, numerical examples that provide insights in the system performance are given in Section \ref{sec:num}.

\section{System Model} \label{sec:model}

\subsection{Network Model}
We consider a slotted random access system consisting of $N=2,3$ users communicating with a common receiver. Each user has an infinite capacity buffer, in which stores arriving and backlogged packets. Packets have equal length and the time is divided into slots corresponding to the transmission time of a packet. Let $\{A_{k,n}\}_{n\in\mathbf{N}^{*}}$ be a sequence of independent and identically distributed random variables where $A_{k,n}$ is the number of packets arriving in user node $k$, $k=1,2$, in the time interval $(n,n+1]$, with $E(A_{k,n})=\lambda_{k}<\infty$, $k=1,2$. Denote also by $D(x,y)=\lim_{n\to\infty}E(x^{A_{1,n}}y^{A_{2,n}})$, $|x|\leq1$, $|y|\leq1$, $n\in\mathbf{N}^{*}$, the generating function of the joint distribution of the number of arrivals in any slot.

At the beginning of each slot, there is a probability for the node $k$, $k=1,2$, to transmit a packet to the receiver. More than one concurrent transmission can occur without having a collision.

Due to the interference and the complex interdependence among the nodes we consider the following queue-aware transmission policy: If both nodes are non empty, node $k$, $k=1,2,$ transmits a packet with probability $\alpha_{k}$ independently, where $\bar{\alpha}_{k}=1-\alpha_{k}$ is the probability that node $k$ does not make a transmission in a slot when its queue is not empty. If node 1 (resp. 2) is the only non-empty, it transmits a packet with probability $\alpha_{k}^{*}$, where $\bar{\alpha}_{k}^{*}=1-\alpha_{k}^{*}$ is the probability that node $k$ does not make a transmission in the given slot. \footnote{We consider the general case for $\alpha_{k}^{*}$ instead of assuming directly $\alpha_{k}^{*}=1$. This can handle cases where the node cannot transmit with probability one even if the other node is silent. This scenario for example can occur when the nodes are subject to energy limitations. It is outside of the scope of this work to consider specific reasons when this case can appear but we want to keep the proposed analysis general.} Note that in our case, a node is aware about the state of its neighbor. \footnote{In a shared access network, it is practical to assume some minimum exchanging information of one bit in this case.}

\subsection{Physical Layer Model}

The MPR (Multi-Packet Reception) channel model used in this paper is a generalized form of the packet erasure model. In the wireless environment, a packet can be decoded correctly by the receiver if the received SINR exceeds a certain threshold. More precisely, suppose that we are given a set $T$ of nodes transmitting in the same time slot. Let $P_{rx}(i,j)$ be the signal
power received from node $i$ at node $j$ (when $i$ transmits), and let $SINR(i,j)$ be the SINR received by node j, i.e., $SINR(i,j)=\frac{P_{rx}(i,j)}{n_{j}+\sum_{k\in T-\{i\}}P_{rx}(k,j)}$, 
where $n_{j}$ denotes the receiver noise power at $j$. We assume that a packet transmitted by $i$ is successfully received by $j$ if and only if $SINR(i,j)\geq \gamma_i$, where $\gamma_i$ is the SINR threshold. The wireless channel is subject to fading; let $P_{tx}(i)$ be the transmitting power at node $i$ and $r(i,j)$ be the distance between $i$ and $j$. The power received by $j$ when $i$ transmits is $P_{rx}(i,j) = A(i,j)g(i,j)$ where $A(i,j)$ is a random variable representing channel fading. We assume that the fading model is slow, flat fading, constant during a time slot and independently varying from time slot to time slot. Under Rayleigh fading, it is known \cite{tse} that $A(i,j)$ is exponentially distributed. The received power factor $g(i,j)=P_{tx}(i)(r(i,j))^{-h}$ where $h$ is the path loss exponent with typical values between $2$ and $6$. In this study we consider one destination which is common for both nodes, thus $j$ denotes the common destination here and we can also write $SINR(i,j) = SINR_i$.
The success probability of link $i,j$ when the transmitting nodes are in $T$ is given by \cite{tse}
\begin{equation}
\begin{array}{c}
P_{s}(i,T)=\mathrm{Pr}\left(SINR_i \geq \gamma_i\right)=\exp \left(\frac{-\gamma_{i}n_{j}}{v(i,j)g(i,j)}\right)\prod_{k\in T-\{i\}} \left(1+\frac{v(k,j)g(k,j)}{v(i,j)g(i,j)}\right)^{-1},
\end{array}
\label{pro}
\end{equation}
where $v(i,j)$ is the parameter of the Rayleigh random variable for fading.
According to (\ref{pro}) we denote $P_{i/\{i,j\}}$ to be the success probability of node $i$ when the transmitting nodes are $i$ and $j$, $i,j=1,2$. More precisely: the strongest user can be successfully received even in the presence of simultaneous transmissions (i.e., collision), if the difference in power is large enough \cite{clau_tcom1990} (provided that $SINR(i,k)\geq\gamma_{k}$). If both nodes transmit, but their $SINR$ are below the threshold $\gamma_{k}$, their transmission is unsuccessful. 

Next, we will define for convenience some conditional probabilities on top of the expression given in \eqref{pro}.\footnote{A similar approach can be found in \cite{NawareTong2005,WangTSP2007}.} We define $P_{1/\{1,2\}}$ the probability that when both nodes $1$ and $2$ are transmitting only the transmission from node $1$ is successful. Then $P_{1/\{1,2\}}=\mathrm{Pr}\left(SINR_1 \geq \gamma_1,SINR_2 < \gamma_2 \right)$. Similarly we can define $P_{2/\{1,2\}}$. The $P_{1,2/\{1,2\}}$ is the probability that both packets transmitted by nodes $1$ and $2$ are transmitted successfully, then $P_{1,2/\{1,2\}}=\mathrm{Pr}\left(SINR_1 \geq \gamma_1,SINR_2 \geq \gamma_2 \right)$. Thus, $P_{s}(1,\{1,2\})= \mathrm{Pr}\left(SINR_1 \geq \gamma_1\right)= \mathrm{Pr}\left(SINR_1 \geq \gamma_1,SINR_2 < \gamma_2 \right)+\mathrm{Pr}\left(SINR_1 \geq \gamma_1,SINR_2 \geq \gamma_2 \right)$, then we have that $P_{s}(1,\{1,2\})=P_{1/\{1,2\}}+P_{1,2/\{1,2\}}$.

The $P_{0/\{1,2\}}=\mathrm{Pr}\left(SINR_1 < \gamma_1,SINR_2 < \gamma_2 \right)$ is the probability where both packets fail to reach the destination when both nodes $1$ and $2$ are transmitting, then $P_{0/\{1,2\}}=1-P_{1/\{1,2\}}-P_{2/\{1,2\}}-P_{1,2/\{1,2\}}$. Note that $P_{i/\{i\}}=P_{s}(i,\{i\})$ is the success probability of node $i$ when only $i$-th node transmits but the other one is active (i.e., there are packets stored in its buffer), we denote with $P_{0/\{i\}}=1-P_{s}(i,\{i\})$ the outage probability respectively. Furthermore, we assume that a node adjusts its transmission parameters such as the transmission power when the other node has an empty queue (i.e is inactive). Thus, the success (resp. outage) probability of node $i$ when the other node is inactive is denoted by $\tilde{P}_{i/\{i\}}$ (resp. $\tilde{P}_{0/\{i\}}$). 
By allowing this we can consider a simple power control policy where a node can adapt its transmission power when the other node is empty, in order to increase the success probability thus, is reasonable to assume that $\tilde{P}_{i/\{i\}}\geq P_{i/\{i\}}$.

In the case of an unsuccessful transmission the packet has to be re-transmitted later. We assume that the receiver gives an instantaneous (error-free) feedback (ACK) of all the packets that were successful in a slot at the end of the slot to all the nodes. The nodes remove the successfully transmitted packets from their buffers. The packets that were not successfully transmitted are retained.

Now we can write the expressions for the average service rates $\mu_1$ and $\mu_2$ seen at node $1$ and $2$ respectively. The expression for $\mu_1$ is given below and similarly we can obtain $\mu_2$. Denote by $N_{k,n}$ the length of queue $k$ at the beginning of time slot $n$. Then,
\begin{equation} \label{eq:mu1}
\begin{array}{c}
\mu_1=Pr(N_2 \neq 0) \left[\alpha_1\bar{\alpha}_{2}P_{1/\{1\}}+\alpha_1 \alpha_2 \left(P_{1/\{1,2\}}+P_{1,2/\{1,2\}}\right)\right]  +Pr(N_2 = 0) \alpha_{1}^{*} \tilde{P}_{1/{1}},
\end{array}
\end{equation}
where $N_{k}=\lim_{n\to\infty}N_{k,n}$, $k=1,2$. We can easily see from (\ref{eq:mu1}) that the service rate of one queue depends on the status of the other queue. Thus, the queues are coupled. In Section \ref{sec:stability2} we bypass this difficulty by applying the stochastic dominance technique to obtain the exact stability region. Regarding the delay analysis we need a different treatment, based on the powerful and technical theory on boundary value problems; see Section \ref{sec:bound}.

Based on the definition in~\cite{szpa}, a user's node is said to be \emph{stable} if
$\lim_{n \rightarrow \infty} {Pr}[N_{k,n}< {x}] = F(x)$ and $\lim_{ {x} \rightarrow \infty} F(x) = 1$.
Loynes' theorem~\cite{b:Loynes} states that if the arrival and service processes of a queue are strictly jointly stationary and the average arrival rate is less than the average service rate, then the queue is stable. If the average arrival rate is greater than the average service rate, then the queue is unstable and the value of $N_{k,n}$ approaches infinity almost surely. The stability region of the system is defined as the set of arrival rate vectors $\boldsymbol{\lambda}=(\lambda_1, \lambda_2)$, for which the queues in the system are stable.

\section{Stability Region for $N=2$ users} \label{sec:stability2}

In this section, we derive the stability region, i.e., the region of values for $\lambda_{k}$, $k=1,2$, for which our system is stable. The following theorem provides the stability conditions for the two-user random access network.

\begin{thm}\label{lem}
The stability region $\mathcal{R}$ for a fixed transmission probability vector $\mathbf{p}:=[\alpha_{1},\alpha_{1}^{*},\alpha_{2},\alpha_{2}^{*}]$ is given by 
$\mathcal{R}=\mathcal{R}_1 \cup \mathcal{R}_2$ where 
\begin{align} \label{stab_R1}
\mathcal{R}_1 = \left\{ (\lambda_{1},\lambda_{2}) :\lambda_{1}<\alpha_{1}^{*}\tilde{P}_{1/\{1\}}+\widehat{d}_{1}\frac{\lambda_{2}}{\alpha_{2}\widehat{\alpha}_{1}},\lambda_{2}<\alpha_{2}\widehat{\alpha}_{1} \right\},\\ 
\mathcal{R}_2 = \left\{ (\lambda_{1},\lambda_{2}) :\lambda_{2}<\alpha_{2}^{*}\tilde{P}_{2/\{2\}}+\widehat{d}_{2}\frac{\lambda_{1}}{\alpha_{1}\widehat{\alpha}_{2}},\lambda_{1}<\alpha_{1}\widehat{\alpha}_{2} \right\},
\end{align}
where $\widehat{d}_{k}=d_{k}+\alpha_{1}\alpha_{2}P_{1,2/\{1,2\}}$ for $k=1,2$, $d_{1}=\alpha_{1}(\bar{\alpha}_{2}P_{1/\{1\}}+\alpha_{2}P_{1/\{1,2\}})-\alpha_{1}^{*}\tilde{P}_{1/\{1\}}$, $d_{2}=\alpha_{2}(\bar{\alpha}_{1}P_{2/\{2\}}+\alpha_{1}P_{2/\{1,2\}})-\alpha_{2}^{*}\tilde{P}_{2/\{2\}}$, $\widehat{\alpha}_{1}=\bar{\alpha}_{1}P_{2/\{2\}}+\alpha_{1}(P_{2/\{1,2\}}+P_{1,2/\{1,2\}})$, $\widehat{\alpha}_{2}=\bar{\alpha}_{2}P_{1/\{1\}}+\alpha_{2}(P_{1/\{1,2\}}+P_{1,2/\{1,2\}})$.
\end{thm}

\begin{proof}
To determine the stability region of our system (depicted in Fig. \ref{fig:region}) we apply the stochastic dominance technique \cite{Rao_TIT1988}, i.e. we construct hypothetical dominant systems, in which the source transmits dummy packets for the packet queue that is empty, while for the non-empty queue it transmits according to its traffic. Under this approach, we consider the $R_{1}$, and $R_{2}$-dominant systems. In the $R_{k}$ dominant system, whenever the queue of user $k$, $k=1,2$ empties, it continues transmitting a dummy packet.

Thus, in $R_{1}$, node $1$ never empties, and hence, node $2$ sees a constant service rate, while the service rate of node $1$ depends on the state of node $2$, i.e., empty or not. We proceed with queue at node $1$. The service rate of the first node is given by \eqref{eq:mu1}. The service rate of the second user is given by
\begin{equation} 
\begin{array}{c}
\mu_2=\alpha_2\bar{\alpha}_{1}P_{2/\{2\}}+\alpha_2 \alpha_1 \left(P_{2/\{1,2\}}+P_{1,2/\{1,2\}}\right).
\end{array}
\end{equation}
By applying Loyne's criterion, the second node is stable if and only if the average arrival rate is less that the average service rate, $\lambda_{2} < \alpha_2\bar{\alpha}_{1}P_{2/\{2\}}+\alpha_2 \alpha_1 \left(P_{2/\{1,2\}}+P_{1,2/\{1,2\}}\right)$. We can obtain the probability that the second node is empty and is given by $Pr(N_2 = 0)=1-\frac{\lambda_{2}}{\mu_2}$. After replacing $Pr(N_2 = 0)$ into \eqref{eq:mu1}, and applying Loynes criterion we can obtain the stability condition for the first node. Then, we have the stability region $\mathcal{R}_1$ given by \eqref{stab_R1}. Note that the expression in \eqref{stab_R1} is given in a more compact form that it will be useful in the next sections.
Similarly, we can obtain the stability region for the second dominant system $\mathcal{R}_2$, the proof is omitted due to space limitations. For a detailed treatment of dominant systems please refer to \cite{Rao_TIT1988}.

An important observation made in \cite{Rao_TIT1988} is that the stability conditions obtained by the stochastic dominance technique are not only sufficient but also necessary for the stability of the original system. The \emph{indistinguishability} argument \cite{Rao_TIT1988} applies here as well. Based on the construction of the dominant system, we can see that the queue sizes in the dominant system are always greater than those in the original system, provided they are both initialized to the same value and the arrivals are identical in both systems. Therefore, given $\lambda_{2}<\mu_{2}$, if for some $\lambda_{1}$, the queue at the first user is stable in the dominant system, then the corresponding queue in the original system must be stable. Conversely, if for some $\lambda_{1}$ in the dominant system, the queue at the first node saturates, then it will not transmit dummy packets, and as long as the first user has a packet to transmit, the behavior of the dominant system is identical to that of the original system since dummy packet transmissions are eliminated as we approach the stability boundary. Therefore, the original and the dominant system are indistinguishable at the boundary points.
\end{proof}

\begin{remark}
The stability region is a convex polyhedron when the following condition holds
$\frac{\alpha_{1}\widehat{\alpha}_{2}}{\alpha_{1}^{*}\tilde{P}_{1/\{1\}}}+\frac{\alpha_{2}\widehat{\alpha}_{2}}{\alpha_{2}^{*}\tilde{P}_{2/\{2\}}}\geq1$. 
When equality holds in the previous condition as depicted also in Fig. \ref{fig:region}, the region is a triangle and coincides with the case of time-sharing. Convexity is an important property since it corresponds to the case when parallel concurrent transmissions are preferable to a time-sharing scheme. Additionally, convexity of the stability region implies that if two rate pairs are stable, then any rate pair lying on the line segment joining those two rate pairs is also stable.
\end{remark}

\begin{remark}
The condition
$\frac{\alpha_{1}\widehat{\alpha}_{2}}{\alpha_{1}^{*}\tilde{P}_{1/\{1\}}}+\frac{\alpha_{2}\widehat{\alpha}_{2}}{\alpha_{2}^{*}\tilde{P}_{2/\{2\}}}\geq1$ is the generalized version of the condition that characterizes the MPR capability in the system which was first appeared in \cite{NawareTong2005}.
\end{remark}

\begin{figure}[ht]
	\centering
	\includegraphics[scale=0.3]{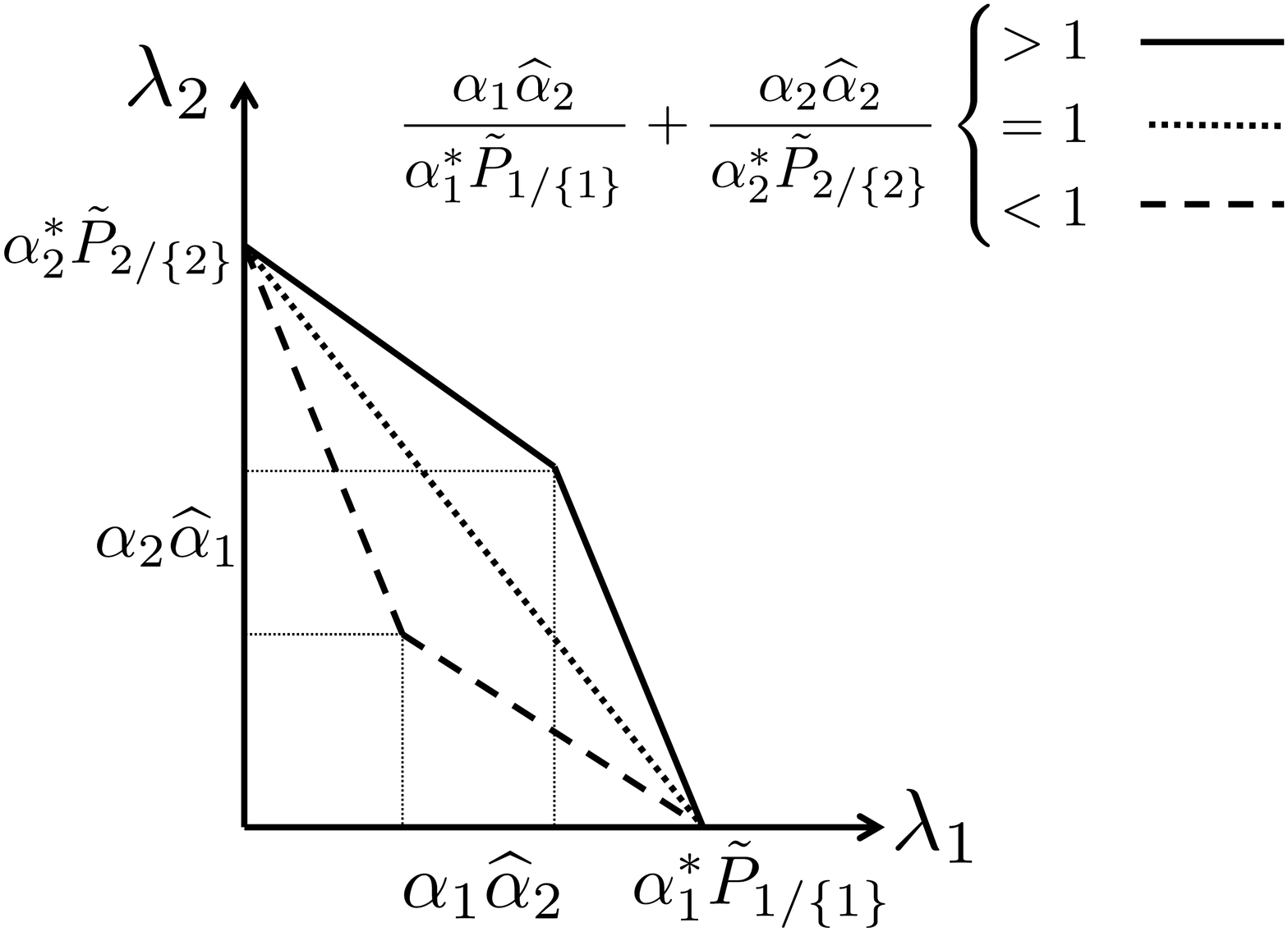}
	\caption{The Stability Region described in Theorem \ref{lem}.}
	\centering
	\label{fig:region}
\end{figure}

\section{Preparatory Analysis \& Results}\label{sec:analysis}

In this section we provide the first part of the analysis that is needed to obtain the expressions for the delay analysis. More explicitly, we derive the fundamental functional equation and we obtain some important results that we will use in the analysis of Section \ref{sec:bound}.

Let $N_{k,n}$ be the number of packets in user node $k$ at the beginning of the $n$th slot. Then, $Y_{n}=(N_{1,n},N_{2,n})$ is a two-dimensional Markov chain with state space is $E=\{(i,j):i,j=0,1,2,...\}$ describing our system model. The queues of both users evolve as $N_{k,n+1}=[N_{k,n}-D_{k,n}]^{+}+A_{k,n},k=1,2$,
where $D_{k,n}$ is the number of departures from user $k$ queue at time slot $n$. Then, the queue evolution equation implies
\begin{equation}
\begin{array}{l}
E(x^{N_{1,n+1}}y^{N_{2,n+1}})=D(x,y)\{P(N_{1,n}=N_{2,n}=0)+E(x^{N_{1,n}}1_{\{N_{1,n}>0,N_{2,n}=0\}})\\\times[\alpha_{1}^{*}\tilde{P}_{0/\{1\}}+\frac{\alpha_{1}^{*}\tilde{P}_{1/\{1\}}}{x}+\bar{\alpha}_{1}^{*}]+E(y^{N_{2,n}}1_{\{N_{1,n}=0,N_{2,n+1}>0\}})[\alpha_{2}^{*}\tilde{P}_{0/\{2\}}+\frac{\alpha_{2}^{*}\tilde{P}_{2/\{2\}}}{y}+\bar{\alpha}_{2}^{*}]\\
+E(x^{N_{1,n}}y^{N_{2,n}}1_{\{N_{1,n}>0,N_{2,n}>0\}})[\frac{\alpha_{1}\bar{\alpha}_{2}P_{1/\{1\}}}{x}+\frac{\alpha_{2}\bar{\alpha}_{1}P_{2/\{2\}}}{y}+\alpha_{1}\bar{\alpha}_{2}P_{0/\{1\}}+\alpha_{2}\bar{\alpha}_{1}P_{0/\{2\}}\\
+\bar{\alpha}_{1}\bar{\alpha}_{2}+\alpha_{1}\alpha_{2}(P_{0/\{1,2\}}+\frac{P_{1/\{1,2\}}}{x}+\frac{P_{2/\{1,2\}}}{y}+\frac{P_{1,2/\{1,2\}}}{xy})]\},
\end{array}
\label{dxx}
\end{equation}
where $1_{\{S\}}$ denotes the indicator function of the event $S$. Assuming that the system is stable, let $H(x,y)=\lim_{n\to\infty}E(x^{N_{1,n}}y^{N_{2,n}}),\,|x|\leq1,\,|y|\leq1$.
Using (\ref{dxx}) after some algebra we obtain the following functional equation,
\begin{equation}
\begin{array}{c}
R(x,y)H(x,y)=A(x,y)H(x,0)+B(x,y)H(0,y)+C(x,y)H(0,0),
\end{array}
\label{we}
\end{equation}
where
\begin{displaymath}
\begin{array}{rl}
R(x,y)=&D^{-1}(x,y)-1+\alpha_{1}(\bar{\alpha}_{2}P_{1/\{1\}}+\alpha_{2}P_{1/\{1,2\}})(1-\frac{1}{x})\vspace{2mm}\\
&+\alpha_{2}(\bar{\alpha}_{1}P_{2/\{2\}}+\alpha_{1}P_{2/\{1,2\}})(1-\frac{1}{y})+\alpha_{1}\alpha_{2}P_{1,2/\{1,2\}}(1-\frac{1}{xy}),\vspace{2mm}\\
A(x,y)=&\alpha_{2}(\bar{\alpha}_{1}P_{2/\{2\}}+\alpha_{1}P_{2/\{1,2\}})(1-\frac{1}{y})+d_{1}(1-\frac{1}{x})+\alpha_{1}\alpha_{2}P_{1,2/\{1,2\}}(1-\frac{1}{xy}),\vspace{2mm}\\
B(x,y)=&\alpha_{1}(\bar{\alpha}_{2}P_{1/\{1\}}+\alpha_{2}P_{1/\{1,2\}})(1-\frac{1}{x})+d_{2}(1-\frac{1}{y})+\alpha_{1}\alpha_{2}P_{1,2/\{1,2\}}(1-\frac{1}{xy}),\vspace{2mm}\\
C(x,y)=&d_{2}(\frac{1}{y}-1)+d_{1}(\frac{1}{x}-1)+\alpha_{1}\alpha_{2}P_{1,2/\{1,2\}}(\frac{1}{xy}-1).
\end{array}
\end{displaymath}

In the following, we will solve (\ref{we}) by assuming geometrically distributed arrival processes at both stations, which are assumed to be independent. More precisely, we assume $D(x,y)=[(1+\lambda_{1}(1-x))(1+\lambda_{2}(1-y))]^{-1}$, $|x|\leq1$, $|y|\leq1$.

\begin{remark}
Without loss of generality, it is realistic to assume that $d_{i}<0$, $i=1,2$. Indeed, $d_{i}$ is the difference of the successful transmission probability of node $i$, when both nodes are active (i.e., both nodes have packets to send) minus the successful transmission probability of node $i$ when the other node is inactive (i.e., only node $i$ has packets to send). Clearly, in the later case due to the lack of interference and since node $i$ senses the other node inactive will transmit with a higher probability in order to exploit the idle slot of the other node. Definitely, in such a case $\tilde{P}_{i/\{i\}}\geq P_{i/\{i\}}\geq P_{i/\{1,2\}}$.\footnote{The particular case where $\tilde{P}_{i/\{i\}}= P_{i/\{i\}}= P_{i/\{1,2\}}$ is omitted since there is no coupling between the queues and the analysis becomes trivial.} Therefore from hereon we assume that $d_{i}<0$, $i=1,2$.
\end{remark}

Some interesting relations can be obtained directly from the functional equation (\ref{we}). Taking $y = 1$, dividing by $x-1$ and taking $x\to 1$ in (\ref{we}) and vice versa yield the following ``conservation of flow" relations:
\begin{equation}
\lambda_{1}=\alpha_{1}\widehat{\alpha}_{2}(1-H(0,1)-H(1,0)+H(0,0))+\alpha_{1}^{*}\tilde{P}_{1/\{1\}}(H(1,0)-H(0,0)),
\label{r1}
\end{equation}
\begin{equation}
\lambda_{2}=\alpha_{2}\widehat{\alpha}_{1}(1-H(0,1)-H(1,0)+H(0,0))+\alpha_{2}^{*}\tilde{P}_{2/\{2\}}(H(0,1)-H(0,0)),
\label{r2}
\end{equation}
where $\widehat{\alpha}_{1}=\bar{\alpha}_{1}P_{2/\{2\}}+\alpha_{1}(P_{2/\{1,2\}}+P_{1,2/\{1,2\}})$, $\widehat{\alpha}_{2}=\bar{\alpha}_{2}P_{1/\{1\}}+\alpha_{2}(P_{1/\{1,2\}}+P_{1,2/\{1,2\}})$.

In the following, the analysis is distinguished in two cases:
\begin{enumerate}
\item For $\frac{\alpha_{1}\widehat{\alpha}_{2}}{\alpha_{1}^{*}\tilde{P}_{1/\{1\}}}+\frac{\alpha_{2}\widehat{\alpha}_{1}}{\alpha_{2}^{*}\tilde{P}_{2/\{2\}}}=1$, equations (\ref{r1}), (\ref{r2}) yield $H(0,0)=1-\frac{\lambda_{1}}{\alpha_{1}^{*}\tilde{P}_{1/\{1\}}}-\frac{\lambda_{2}}{\alpha_{2}^{*}\tilde{P}_{2/\{2\}}}=1-\rho$.
\item In case $\frac{\alpha_{1}\widehat{\alpha}_{2}}{\alpha_{1}^{*}\tilde{P}_{1/\{1\}}}+\frac{\alpha_{2}\widehat{\alpha}_{1}}{\alpha_{2}^{*}\tilde{P}_{2/\{2\}}}\neq1$, equations (\ref{r1}), (\ref{r2}) yield
\begin{equation}\label{rd}
\begin{array}{lr}
H(1,0)=\frac{\alpha_{1}\widehat{\alpha}_{2}(\lambda_{2}-\alpha_{2}^{*}\tilde{P}_{2/\{2\}})-\lambda_{1}\widehat{d}_{2}-\alpha_{1}^{*}\tilde{P}_{1/\{1\}}\widehat{d}_{2}H(0,0)}{d_{1}d_{2}-\alpha_{1}\widehat{\alpha}_{2}\alpha_{2}\widehat{\alpha}_{1}},&
H(0,1)=\frac{\alpha_{2}\widehat{\alpha}_{1}(\lambda_{1}-\alpha_{1}^{*}\tilde{P}_{1/\{1\}})-\lambda_{2}\widehat{d}_{1}-\alpha_{2}^{*}\tilde{P}_{2/\{2\}}\widehat{d}_{1}H(0,0)}{{d_{1}d_{2}-\alpha_{1}\widehat{\alpha}_{2}\alpha_{2}\widehat{\alpha}_{1}}}.
\end{array}
\end{equation}
\end{enumerate}

We now focus on the derivation of some preparatory results in view of the resolution of functional equation (\ref{we}). More precisely, we focus on the analysis of the kernel equation $R(x,y)=0$.
\subsection{Analysis of the kernel}\label{ker}
In the following we consider the kernel equation $R(x,y)=0$ and provide some important properties. We focus on a subclass of MPR channels, the so called ``capture" channels, i.e., $P_{1,2/\{1,2\}}=0$ (at most one user has a successful packet transmission even if many users transmit in that slot \cite{capture1,capture2}). Note that, $R(x,y)=a(x)y^{2}+b(x)y+c(x)=\widehat{a}(y)x^{2}+\widehat{b}(y)x+\widehat{c}(y)$,
where, $a(x)=\lambda_{2}x(\lambda_{1}(x-1)-1),\,b(x)=x(\lambda+\lambda_{1}\lambda_{2}+\alpha_{1}\widehat{\alpha}_{2}+\alpha_{2}\widehat{\alpha}_{1})-\alpha_{1}\widehat{\alpha}_{2}-\lambda_{1}(1+\lambda_{2})x^{2}$, $c(x)=-\alpha_{2}\widehat{\alpha}_{1}x$, $\widehat{a}(y)=\lambda_{1}y(\lambda_{2}(y-1)-1)$, $\widehat{b}(y)=y(\lambda+\lambda_{1}\lambda_{2}+\alpha_{1}\widehat{\alpha}_{2}+\alpha_{2}\widehat{\alpha}_{1})-\alpha_{2}\widehat{\alpha}_{1}-\lambda_{2}(1+\lambda_{1})y^{2}$, $\widehat{c}(y)=-\alpha_{1}\widehat{\alpha}_{2}y$.
The roots of $R(x,y)=0$ are $X_{\pm}(y)=\frac{-\widehat{b}(y)\pm\sqrt{D_{y}(y)}}{2\widehat{a}(y)}$, $Y_{\pm}(x)=\frac{-b(x)\pm\sqrt{D_{x}(x)}}{2a(x)}$, where $D_{y}(y)=\widehat{b}(y)^{2}-4\widehat{a}(y)\widehat{c}(y)$, $D_{x}(x)=b(x)^{2}-4a(x)c(x)$.

\begin{lemma}
For $|y|=1$, $y\neq1$, the kernel equation $R(x,y)=0$ has exactly one root $x=X_{0}(y)$ such that $|X_{0}(y)|<1$. For $\lambda_{1}<\alpha_{1}\widehat{\alpha}_{2}$, $X_{0}(1)=1$. Similarly, we can prove that $R(x,y)=0$ has exactly one root $y=Y_{0}(x)$, such that $|Y_{0}(x)|\leq1$, for $|x|=1$.
\end{lemma}

\begin{proof}
It is easily seen that $R(x,y)=\frac{xy-\Psi(x,y)}{xyD(x,y)}$, where $\Psi(x,y)=D(x,y)[xy-y(x-1)\alpha_{1}\widehat{\alpha}_{2}-x(y-1)\alpha_{2}\widehat{\alpha}_{1}]$, where for $|x|\leq1$, $|y|\leq1$, $\Psi(x,y)$ is a generating function of a proper probability distribution. Now, for $|y|=1$, $y\neq1$ and $|x|=1$ it is clear that $|\Psi(x,y)|<1=|xy|$. Thus, from Rouch\'e's theorem, $xy-\Psi(x,y)$ has exactly one zero inside the unit circle. Therefore, $R(x,y)=0$ has exactly one root $x=X_{0}(y)$, such that $|x|<1$. For $y=1$, $R(x,1)=0$ implies $(x-1)\left(\lambda_{1}-\frac{\alpha_{1}\widehat{\alpha}_{2}}{x}\right)=0$.
Therefore, for $y=1$, and since $\lambda_{1}<\alpha_{1}\widehat{\alpha}_{2}$, the only root of $R(x,1)=0$ for $|x|\leq1$, is $x=1$.
\end{proof}

\begin{lemma}\label{lem1}
The algebraic function $Y(x)$, defined by $R(x,Y(x)) = 0$, has four real branch points $0< x_{1}<x_{2}\leq1<x_{3}<x_{4}<\frac{1+\lambda_{1}}{\lambda_{1}}$. Moreover, $D_{x}(x)<0$, $x\in(x_{1},x_{2})\cup(x_{3},x_{4})$ and $D_{x}(x)<0$, $x\in(-\infty,x_{1})\cup(x_{2},x_{3})\cup(x_{4},\infty)$. Similarly, $X(y)$, defined by $R(X(y),y) = 0$, has four real branch points $0\leq y_{1}<y_{2}\leq1<y_{3}<y_{4}<\frac{1+\lambda_{2}}{\lambda_{2}}$. Moreover, $D_{x}(y)<0$, $y\in(y_{1},y_{2})\cup(y_{3},y_{4})<$ and $D_{x}(y).0$, $y\in(-\infty,y_{1})\cup(y_{2},y_{3})\cup(y_{4},\infty)$.
\end{lemma}
\begin{proof}
The proof is based on simple algebraic arguments and further details are omitted due to space limitations.
\end{proof}

To ensure the continuity of the function two valued function $Y(x)$ (resp. $X(y)$) we consider the following cut planes: $\doubletilde{C}_{x}=C_{x}-([x_{1},x_{2}]\cup[x_{3},x_{4}]$, $\doubletilde{C}_{y}=C_{y}-([y_{1},y_{2}]\cup[y_{3},y_{4}]$, where $C_{x}$, $C_{y}$ the complex planes of $x$, $y$, respectively. In $\doubletilde{C}_{x}$ (resp. $\doubletilde{C}_{y}$), denote by $Y_{0}(x)$ (resp. $X_{0}(y)$) the zero of $R(x,Y(x))=0$ (resp. $R(X(y),y)=0$) with the smallest modulus, and $Y_{1}(x)$ (resp. $X_{1}(y)$) the other one. 

Define the image contours, $\mathcal{L}=Y_{0}[\overrightarrow{\underleftarrow{x_{1},x_{2}}}]$, $\mathcal{M}=X_{0}[\overrightarrow{\underleftarrow{y_{1},y_{2}}}]$, where $[\overrightarrow{\underleftarrow{u,v}}]$ stands for the contour traversed from $u$ to $v$ along the upper edge of the slit $[u,v]$ and then back to $u$ along the lower edge of the slit. 
The following lemma shows that the mappings $Y(x)$, $X(y)$, for $x\in[x_{1},x_{2}]$, $y\in[y_{1},y_{2}]$ respectively, give rise to the smooth and closed contours $\mathcal{L}$, $\mathcal{M}$ respectively.
\begin{lemma}\label{sq}\begin{enumerate}\item For $y\in[y_{1},y_{2}]$, the algebraic function $X(y)$ lies on a closed contour $\mathcal{M}$, which is symmetric with respect to the real line and defined by $|x|^{2}=m(Re(x))$, $m(\delta)=\frac{\alpha_{1}\widehat{\alpha}_{2}}{\lambda_{1}(1+\lambda_{2}-\lambda_{2}\zeta(\delta))}$, and $|x|^{2}\leq\frac{\alpha_{1}\widehat{\alpha}_{2}}{\lambda_{1}(1+\lambda_{2}-\lambda_{2}y_{2})}$,
where, $\zeta(\delta)=\frac{k(\delta)-\sqrt{k^{2}(\delta)-4\alpha_{2}\widehat{\alpha}_{1}(\lambda_{2}(1+\lambda_{1}(1-2\delta)))}}{2\lambda_{2}(1+\lambda_{1}(1-2\delta))}$, $k(\delta):=\lambda+\lambda_{1}\lambda_{2}+\alpha_{1}\widehat{\alpha}_{2}+\alpha_{2}\widehat{\alpha}_{1}-2\lambda_{1}(1+\lambda_{2})\delta$. 
Set $\beta_{0}:=\sqrt{\frac{\alpha_{1}\widehat{\alpha}_{2}}{\lambda_{1}(1+\lambda_{2}-\lambda_{2}y_{2})}}$, $\beta_{1}=-\sqrt{\frac{\alpha_{1}\widehat{\alpha}_{2}}{\lambda_{1}(1+\lambda_{2}-\lambda_{2}y_{1})}}$ the extreme right and left point of $\mathcal{M}$, respectively.
\item For $x\in[x_{1},x_{2}]$, the algebraic function $Y(x)$ lies on a closed contour $\mathcal{L}$, which is symmetric with respect to the real line and defined by $|y|^{2}=v(Re(y)),\,v(\delta)=\frac{\alpha_{2}\widehat{\alpha}_{1}}{\lambda_{2}(1+\lambda_{1}-\lambda_{1}\theta(\delta))}$, $|y|^{2}\leq\frac{\alpha_{2}\widehat{\alpha}_{1}}{\lambda_{2}(1+\lambda_{1}-\lambda_{1}x_{2})}$,
where $\theta(\delta)=\frac{l(\delta)-\sqrt{l^{2}(\delta)-4\alpha_{1}\widehat{\alpha}_{2}(\lambda_{1}(1+\lambda_{2}(1-2\delta)))}}{2\lambda_{1}(1+\lambda_{2}(1-2\delta))}$, $l(\delta):=\lambda+\lambda_{1}\lambda_{2}+\alpha_{1}\widehat{\alpha}_{2}+\alpha_{2}\widehat{\alpha}_{1}-2\lambda_{2}(1+\lambda_{1})\delta$.
Set $\eta_{0}:=\sqrt{\frac{\alpha_{2}\widehat{\alpha}_{1}}{\lambda_{2}(1+\lambda_{1}-\lambda_{1}x_{2})}}$, $\eta_{1}=-\sqrt{\frac{\alpha_{2}\widehat{\alpha}_{1}}{\lambda_{2}(1+\lambda_{1}-\lambda_{1}x_{1})}}$ the extreme right and left point of $\mathcal{L}$, respectively.
\end{enumerate}
\end{lemma}
\begin{proof}
We will prove the part related to $\mathcal{M}$. Similarly, we can also prove part 2. For $y\in[y_{1},y_{2}]$, $D_{y}(y)$ is negative, so $X_{0}(y)$ and $X_{1}(y)$ are complex conjugates. It also follows that
\begin{equation}
\begin{array}{l}
Re(X(y))=\frac{y(\lambda+\lambda_{1}\lambda_{2}+\alpha_{1}\widehat{\alpha}_{2}+\alpha_{2}\widehat{\alpha}_{1})-\alpha_{2}\widehat{\alpha}_{1}-\lambda_{2}(1+\lambda_{1})y^{2}}{2\lambda_{1}y(1+\lambda_{2}-\lambda_{2}y)}.
\end{array}
\label{rd1}
\end{equation} 
Therefore, $|X(y)|^{2}=\frac{\alpha_{1}\widehat{\alpha}_{2}}{\lambda_{1}(1+\lambda_{2}-\lambda_{2}y)}=g(y)$. Clearly, $g(y)$ is an increasing function for $y\in[0,1]$ and thus, $|X(y)|^{2}\leq g(y_{2})$. Finally, $\zeta(\delta)$ is derived by solving (\ref{rd1}) for $y$ with $\delta = Re(X(y))$, and taking the solution such that $y\in[0,1]$.
\end{proof}
\section{The boundary value problems}\label{sec:bound}
As indicated in the previous section, based on a relation between the transmission probabilities of the users, we distinguish the analysis in two cases, which differ both from the modeling and the technical point of view. In this section we consider the case where $P_{1,2/\{1,2\}}=0$.
\subsection{A Dirichlet boundary value problem}
Assume now that $\frac{\alpha_{1}\widehat{\alpha}_{2}}{\alpha_{1}^{*}\tilde{P}_{1/\{1\}}}+\frac{\alpha_{2}\widehat{\alpha}_{1}}{\alpha_{2}^{*}\tilde{P}_{2/\{2\}}}=1$. Then, $A(x,y)=\frac{d_{1}}{\alpha_{1}\widehat{\alpha}_{2}}B(x,y)\Leftrightarrow A(x,y)=\frac{\alpha_{2}\widehat{\alpha}_{1}}{d_{2}}B(x,y)$.
Therefore, for $y\in \mathcal{D}_{y}=\{y\in\mathcal{C}:|y|\leq1,|X_{0}(y)|\leq1\}$,
\begin{equation}
\begin{array}{l}
\alpha_{2}\widehat{\alpha}_{1}H(X_{0}(y),0)+d_{2}H(0,y)+\frac{\alpha_{2}\widehat{\alpha}_{1}C(X_{0}(y),y)}{A(X_{0}(y),y)}(1-\rho)=0.
\end{array}
\label{con}
\end{equation}
For $y\in \mathcal{D}_{y}-[y_{1},y_{2}]$ both $H(X_{0}(y),0)$, $H(0,y)$ are analytic and the right-hand side can be analytically continued up to the slit $[y_{1},y_{2}]$, or equivalently,
\begin{equation}
\begin{array}{c}
\alpha_{2}\widehat{\alpha}_{1}H(x,0)+d_{2}H(0,Y_{0}(x))+\frac{\alpha_{2}\widehat{\alpha}_{1}C(x,Y_{0}(x))}{A(x,Y_{0}(x))}(1-\rho)=0,\,x\in\mathcal{M}.
\end{array}
\label{con2}
\end{equation}
Then, multiplying both sides of (\ref{con2}) by the imaginary complex number $i$, and noticing that $H(0,Y_{0}(x))$ is real for $x\in\mathcal{M}$, since $Y_{0}(x)\in[y_{1},y_{2}]$, we have
\begin{equation}
\begin{array}{c}
Re(iH(x,0))=Re\left(-i\frac{C(x,Y_{0}(x))}{A(x,Y_{0}(x))}\right)(1-\rho),\,x\in\mathcal{M}.
\end{array}
\label{p1}
\end{equation}
Clearly, some analytic continuation considerations must be made in order to have everything well defined. Thus, we have to check for poles of $H(x,0)$ in $S_{x}:=G_{\mathcal{M}}\cap\bar{D}_{x}^{c}$, where $G_{\mathcal{U}}$ be the interior domain bounded by $\mathcal{U}$, and $D_{x}=\{x:|x|<1\}$, $\bar{D}_{x}=\{x:|x|\leq1\}$, $\bar{D}_{x}^{c}=\{x:|x|>1\}$. These poles, if exist, they coincide with the zeros of $A(x,Y_{0}(x))$ in $S_{x}$ (see Appendix). In order to solve (\ref{p1}), we must first conformally transform the problem from $\mathcal{M}$ to the unit circle $\mathcal{C}$. Let the conformal mapping, $z=\gamma(x):G_{\mathcal{M}}\to G_{\mathcal{C}}$, and its inverse $x=\gamma_{0}(z):G_{\mathcal{C}}\to G_{\mathcal{M}}$. 

Then, we have the following problem: Find a function $\tilde{T}(z)=H^{(0)}(\gamma_{0}(z))$ regular for $z\in G_\mathcal{C}$, and continuous for $z\in\mathcal{C}\cup G_\mathcal{C}$ such that, $Re(i\tilde{T}(z))=w(\gamma_{0}(z))$, $z\in\mathcal{C}$. To proceed, we need a representation of $\mathcal{M}$ in polar coordinates, i.e., $\mathcal{M}=\{x:x=\rho(\phi)\exp(i\phi),\phi\in[0,2\pi]\}.$ This procedure is described in detail in \cite{coh}. In the following we summarize the basic steps: Since $0\in G_{\mathcal{M}}$, for each $x\in\mathcal{M}$, a relation between its absolute value and its real part is given by $|x|^{2}=m(Re(x))$ (see Lemma \ref{sq}). Given the angle $\phi$ of some point on $\mathcal{M}$, the real part of this point, say $\delta(\phi)$, is the solution of $\delta-\cos(\phi)\sqrt{m(\delta)}$, $\phi\in[0,2\pi].$ Since $\mathcal{M}$ is a smooth, egg-shaped contour, the solution is unique. Clearly, $\rho(\phi)=\frac{\delta(\phi)}{\cos(\phi)}$, and the parametrization of $\mathcal{M}$ in polar coordinates is fully specified. Then, the mapping from $z\in G_{\mathcal{C}}$ to $x\in G_{\mathcal{M}}$, where $z = e^{i\phi}$ and $x= \rho(\psi(\phi))e^{i\psi(\phi)}$, satisfying $\gamma_{0}(0)=0$ and $\gamma_{0}(z)=\overline{\gamma_{0}(z)}$ is uniquely determined by (see \cite{coh}, Section I.4.4),
\begin{equation}
\begin{array}{rl}
\gamma_{0}(z)=&z\exp[\frac{1}{2\pi}\int_{0}^{2\pi}\log\{\rho(\psi(\omega))\}\frac{e^{i\omega}+z}{e^{i\omega}-z}d\omega],\,|z|<1,\\
\psi(\phi)=&\phi-\int_{0}^{2\pi}\log\{\rho(\psi(\omega))\}\cot(\frac{\omega-\phi}{2})d\omega,\,0\leq\phi\leq 2\pi,
\end{array}
\label{zx}
\end{equation}
i.e., $\psi(.)$ is uniquely determined as the solution of a Theodorsen integral equation with $\psi(\phi)=2\pi-\psi(2\pi-\phi)$. Due to the correspondence-boundaries theorem, $\gamma_{0}(z)$ is continuous in $\mathcal{C}\cup G_{\mathcal{C}}$. 

In the case that $H(x,0)$ has no poles in $S_{x}$, the solution of the problem defined in (\ref{p1}) is:
\begin{equation}
\begin{array}{c}
H(x,0)=-\frac{1-\rho}{2\pi}\int_{|t|=1}f(t)\frac{t+\gamma(x)}{t-\gamma(x)}\frac{dt}{t}+C,\,x\in\mathcal{M},
\end{array}
\label{sol1}
\end{equation}
where $f(t)=Re\left(-i\frac{C(\gamma_{0}(t),Y_{0}(\gamma_{0}(t)))}{A(\gamma_{0}(t),Y_{0}(\gamma_{0}(t)))}\right)$, $C$ is a constant that can be defined by setting $x=0\in G_{\mathcal{M}}$ in (\ref{sol1}) and using the fact that $H(0,0)=1-\rho$, $\gamma(0)=0$. In the case that $H(x,0)$ has a pole, it will be $x=\bar{x}$ (see Appendix), and we still have a Dirichlet problem for the function $(x- \bar{x})H(x,0)$.

Following the discussion above, $C=(1-\rho)\left(1+\frac{1}{2\pi}\int_{|t|=1}f(t)\frac{dt}{t}\right)$.
Setting $t=e^{i\phi}$, $\gamma_{0}(e^{i\phi})=\rho(\psi(\phi))e^{i\psi(\phi)}$, we obtain after some algebra,
\begin{displaymath}
f(e^{i\phi})=\frac{d_{1}\alpha_{2}^{*}\sin(\psi(\phi))(1-Y_{0}(\gamma_{0}(e^{i\phi}))^{-1})}{\rho(\psi(\phi))\left\{\left[\alpha_{2}\widehat{\alpha}_{1}(1-Y_{0}^{-1}(\gamma_{0}(e^{i\phi})))+d_{1}\left(1-\frac{\cos(\psi(\phi))}{\rho(\psi(\phi))}\right)\right]^{2}+\left(d_{1}\frac{\sin(\psi(\phi))}{\rho(\psi(\phi))}\right)^{2}\right\}},
\end{displaymath}
which is an odd function of $\phi$. Thus, $C=1-\rho$. Substituting back in (\ref{sol1}) we deduce after simple calculations 
\begin{equation}
\begin{array}{c}
H(x,0)=(1-\rho)\left\{1+\frac{2\gamma(x)i}{\pi}\int_{0}^{\pi}\frac{f(e^{i\phi})\sin(\phi)}{1-2\gamma(x)\cos(\phi)-\gamma(x)^{2}}\right\},\,x\in G_{\mathcal{M}}.
\end{array}
\label{soll}
\end{equation}
A detailed numerical approach in order to obtain the inverse mapping $\gamma(x)$ is presented in the seminal book \cite{coh}. Similarly, we can determine $H(0,y)$ by solving another Dirichlet boundary value problem on the closed contour $\mathcal{L}$. Then, using the fundamental functional equation (\ref{we}) we uniquely obtain $H(x,y)$.
\subsection{A homogeneous Riemann-Hilbert boundary value problem}
We now assume that $\frac{\alpha_{1}\widehat{\alpha}_{2}}{\alpha_{1}^{*}\tilde{P}_{1/\{1\}}}+\frac{\alpha_{2}\widehat{\alpha}_{1}}{\alpha_{2}^{*}\tilde{P}_{2/\{2\}}}\neq1$. In such a case we consider the following transformation:
\begin{displaymath}
\begin{array}{lr}
G(x):=H(x,0)+\frac{\alpha_{1}^{*}\tilde{P}_{1/\{1\}}d_{2}H(0,0)}{d_{1}d_{2}-\alpha_{1}\widehat{\alpha}_{2}\alpha_{2}\widehat{\alpha}_{1}},&\,
L(y):=H(0,y)+\frac{\alpha_{2}^{*}\tilde{P}_{2/\{2\}}d_{1}H(0,0)}{d_{1}d_{2}-\alpha_{1}\widehat{\alpha}_{2}\alpha_{2}\widehat{\alpha}_{1}}.
\end{array}
\end{displaymath}
Then, for $y\in D_{y}$, (\ref{we}) yields $A(X_{0}(y),y)G(X_{0}(y))=-B(X_{0}(y),y)L(y)$. For $y\in D_{y}-[y_{1},y_{2}]$ both $G(X_{0}(y))$, $L(y)$ are analytic and the right-hand side can be analytically continued up to the slit
$[y_1, y_2]$, or equivalently for $x\in\mathcal{M}$,
\begin{equation}
\begin{array}{c}
A(x,Y_{0}(x))G(x)=-B(x,Y_{0}(x))L(Y_{0}(x)).
\end{array}
\label{za1}
\end{equation}
Clearly, $G(x)$ is holomorphic in $D_{x}$, continuous in $\bar{D}_{x}$. However, $G(x)$ might has poles, based on the values of the system parameters in $S_{x}=G_{\mathcal{M}}\cap\bar{D}_{x}^{c}$. These poles (if exist) coincide with the zeros of $A(x,Y_{0}(x))$ in $S_{x}$; see Appendix. For $y\in[y_{1},y_{2}]$, let $X_{0}(y)=x\in\mathcal{M}$ and realize that $Y_{0}(X_{0}(y))=y$ so that $y=Y_{0}(x)$ (note that following \cite{avr} $B(x,Y_{0}(x))\neq0$, $x\in\mathcal{M}$). Taking into account the possible poles of $G(x)$, and noticing that $L(Y_{0}(x))$ is real for $x\in\mathcal{M}$, since $Y_{0}(x)\in[y_{1},y_{2}]$, we have
\begin{equation}
\begin{array}{c}
Re[iU(x)\tilde{G}(x)]=0,\,x\in\mathcal{M},\vspace{2mm}\\
U(x)=\frac{A(x,Y_{0}(x))}{(x-\bar{x})^{r}B(x,Y_{0}(x))},\,\,\tilde{G}(x)=(x-\bar{x})^{r}G(x),
\end{array}
\label{df3}
\end{equation}
where $r=0,1$, whether $\bar{x}$ is zero or not of $A(x,Y_{0}(x))$ in $S_{x}$. Thus, $\tilde{G}(x)$ is regular for $x\in G_{\mathcal{M}}$, continuous for $x\in\mathcal{M}\cup G_{\mathcal{M}}$, and $U(x)$ is a non-vanishing function on $\mathcal{M}$. We must first conformally transform the problem (\ref{df3}) from $\mathcal{M}$ to the unit circle $\mathcal{C}$. Let the conformal mapping $z=\gamma(x):G_{\mathcal{M}}\to G_{\mathcal{C}}$, and its inverse given by $x=\gamma_{0}(z):G_{\mathcal{C}}\to G_{\mathcal{M}}$. 

Then, the Riemann-Hilbert problem formulated in (\ref{df3}) is reduced to the following: Find a function $F(z):=\tilde{H}(\gamma_{0}(z))$, regular in $G_{\mathcal{C}}$, continuous in $G_{\mathcal{C}}\cup\mathcal{C}$ such that, $Re[iU(\gamma_{0}(z))F(z)]=0,\,z\in\mathcal{C}$.

A crucial step in the solution of the problem defined by (\ref{df3}) is the determination of the index $\chi=\frac{-1}{\pi}[arg\{U(x)\}]_{x\in \mathcal{M}}$, where $[arg\{U(x)\}]_{x\in \mathcal{M}}$, denotes the variation of the argument of the function $U(x)$ as $x$ moves along the closed contour $\mathcal{M}$ in the positive direction, provided that $U(x)\neq0$, $x\in\mathcal{M}$. Following the lines in \cite{fay} we have,
\begin{lemma}\begin{enumerate}
\item If $\lambda_{2}<\alpha_{2}\widehat{\alpha}_{1}$, then $\chi=0$ is equivalent to
\begin{displaymath}
\begin{array}{l}
\frac{d A(x,Y_{0}(x))}{dx}|_{x=1}=\alpha_{2}\widehat{\alpha}_{1}\frac{\alpha_{1}\widehat{\alpha}_{2}-\lambda_{1}}{\lambda_{2}-\alpha_{2}\widehat{\alpha}_{1}}+d_{1}<0\Leftrightarrow\lambda_{1}<\alpha_{1}^{*}\tilde{P}_{1/\{1\}}+\widehat{d}_{1}\frac{\lambda_{2}}{\alpha_{2}\widehat{\alpha}_{1}},\vspace{2mm}\\ \frac{d B(X_{0}(y),y)}{dy}|_{y=1}=\alpha_{1}\widehat{\alpha}_{2}\frac{\alpha_{2}\widehat{\alpha}_{1}-\lambda_{2}}{\lambda_{1}-\alpha_{1}\widehat{\alpha}_{2}}+d_{2}<0\Leftrightarrow\lambda_{2}<\alpha_{2}^{*}\tilde{P}_{2/\{2\}}+\widehat{d}_{2}\frac{\lambda_{1}}{\alpha_{1}\widehat{\alpha}_{2}}.
\end{array}
\end{displaymath}
\item If $\lambda_{2}\geq\alpha_{2}\widehat{\alpha}_{1}$, $\chi=0$ is equivalent to $\frac{d B(X_{0}(y),y)}{dy}|_{y=1}<0\Leftrightarrow \lambda_{2}<\alpha_{2}^{*}\tilde{P}_{2/\{2\}}+\widehat{d}_{2}\frac{\lambda_{1}}{\alpha_{1}\widehat{\alpha}_{2}}$.
\end{enumerate}
\end{lemma}

Therefore, under stability conditions (see Lemma \ref{lem}) the problem defined in (\ref{df3}) has a unique solution given by,
\begin{equation}
\begin{array}{rl}
H(x,0)=&D(x-\bar{x})^{r}\exp\left[\frac{1}{2i\pi}\int_{|t|=1}\frac{\log\{J(t)\}}{t-\gamma(x)}dt\right]-\frac{\alpha_{1}^{*}\tilde{P}_{1/\{1\}}d_{2}H(0,0)}{d_{1}d_{2}-\alpha_{1}\widehat{\alpha}_{2}\alpha_{2}\widehat{\alpha}_{1}},\,x\in G_{\mathcal{M}},
\end{array}
\label{sool1}
\end{equation}
where $D$ is a constant and $J(t)=\frac{\overline{U_{1}(t)}}{U_{1}(t)}$, $U_{1}(t)=U(\gamma_{0}(t))$, $|t|=1$. Setting $x=0$ in (\ref{sool1}) we derive a relation between $D$ and $H(0,0)$. Then, for $x=1\in G_{\mathcal{M}}$, and using the first in (\ref{rd}) we can obtain $D$ and $H(0,0)$. Substituting back in (\ref{sool1}) will give
\begin{equation}
\begin{array}{rl}
H(x,0)=&\frac{\lambda_{1}d_{2}-\alpha_{1}\widehat{\alpha}_{2}(\lambda_{2}-\alpha_{2}^{*}\tilde{P}_{2/\{2\}})}{(\alpha_{1}\widehat{\alpha}_{2}\alpha_{2}\widehat{\alpha}_{1}-d_{1}d_{2})(\bar{x}-1)^{r}}\left\lbrace(\bar{x}-x)^{r}\exp[\frac{\gamma(x)-\gamma(1)}{2i\pi}\int_{|t|=1}\frac{\log\{J(t)\}}{(t-\gamma(x))(t-\gamma(1))}dt]\right.\vspace{2mm}\\
&\left.+\frac{\alpha_{1}^{*}\tilde{P}_{1/\{1\}}d_{2}(\bar{x})^{r}}{\alpha_{1}\widehat{\alpha}_{2}\alpha_{2}^{*}\tilde{P}_{2/\{2\}}}\exp[\frac{-\gamma(1)}{2i\pi}\int_{|t|=1}\frac{\log\{J(t)\}}{t(t-\gamma(1))}dt]\right\rbrace,\,x\in G_{\mathcal{M}}.
\end{array}
\label{sool}
\end{equation}

Similarly, we can determine $H(0,y)$ by solving another Dirichlet boundary value problem on the closed contour $\mathcal{L}$. Then, using the fundamental functional equation (\ref{we}) we uniquely obtain $H(x,y)$.
\subsection{Expected Number of Packets and Average Delay}

In the following we derive formulas for the expected number of packets and the average delay at each user node in steady state, say $M_{i}$ and $D_{i}$, $i=1,2,$ respectively. Denote by $H_{1}(x,y)$, $H_{2}(x,y)$ the derivatives of $H(x,y)$ with respect to $x$ and $y$ respectively. Then, $M_{i}=H_{i}(1,1)$, and using Little's law $D_{i}=H_{i}(1,1)/\lambda_{i}$, $i=1,2$. Using the functional equation (\ref{we}) and (\ref{r1}), (\ref{r2}) we derive
\begin{equation}
\begin{array}{lccr}
M_{1}=\frac{\lambda_{1}+d_{1}H_{1}(1,0)}{\alpha_{1}\widehat{\alpha}_{2}},&&&M_{2}=\frac{\lambda_{2}+d_{2}H_{2}(0,1)}{\alpha_{2}\widehat{\alpha}_{1}}.
\end{array}
\label{perf}
\end{equation} 
We only focus on $M_{1}$, $D_{1}$ (similarly we can obtain $M_{2}$, $D_{2}$). Note that $H_{1}(1,0)$ can be obtained using (\ref{sool}) or (\ref{soll}) depending on the value of $\frac{\alpha_{1}\widehat{\alpha}_{2}}{\alpha_{1}^{*}\tilde{P}_{1/\{1\}}}+\frac{\alpha_{2}\widehat{\alpha}_{1}}{\alpha_{2}^{*}\tilde{P}_{2/\{2\}}}$. For the case $\frac{\alpha_{1}\widehat{\alpha}_{2}}{\alpha_{1}^{*}\tilde{P}_{1/\{1\}}}+\frac{\alpha_{2}\widehat{\alpha}_{1}}{\alpha_{2}^{*}\tilde{P}_{2/\{2\}}}\neq1$, using (\ref{sool}),
\begin{equation}
\begin{array}{rl}
H_{1}(1,0)=&\frac{\lambda_{1}d_{2}+\alpha_{1}\widehat{\alpha}_{2}(\alpha_{2}^{*}\tilde{P}_{2/\{2\}}-\lambda_{2})}{\alpha_{1}\widehat{\alpha}_{2}\alpha_{2}\widehat{\alpha}_{1}-d_{1}d_{2}}\{
\frac{\gamma^{\prime}(1)}{2\pi i}\int_{|t|=1}\frac{\log\{J(t)\}}{(t-\gamma(1))^{2}}dt+\frac{r}{1-\bar{x}}1_{\{r=1\}}\}.
\end{array}
\label{xz}
\end{equation}
Substituting (\ref{xz}) in (\ref{perf}) we obtain $M_{1}$, and dividing with $\lambda_{1}$, the average delay $D_{1}$. Note that the calculation of (\ref{zx}) requires the evaluation of integrals (\ref{zx}), and $\gamma(1)$, $\gamma^{\prime}(1)$. For an efficient numerical procedure see \cite{coh}, Chapter IV.1.

\section{Stability conditions: Extension to the case of $N=3$ users}\label{sec:stabb}
In the following, we provide sufficient and necessary conditions for the the case of $N=3$ users based on \cite{szpa1}. In particular we generalize the results in \cite{szpa1}, by including the effect of capture channel as well as the queue-aware transmission policy. We accomplish this by means of a technique based on three simple observations: isolating a single queue from the system, applying Loynes' stability criteria for such an isolated queue, and using stochastic dominance and mathematical induction to verify the required stationarity assumptions in the Loynes' criterion.  Below, we present
an informal overview of the approach. First of all, we construct a modified system as follows. Let $P = (S,U)$ be a partition of $M=\{1,2,3\}$ such that users in $S\neq M$ operate exactly as in the original model, while users in $U$ are able to send packets even if their buffers are empty (e.g., dummy packets). Note that a system consisting of users
in $S$ forms a smaller copy of the original system with slightly new probabilities of transmissions.
Furthermore, it is easy to see that the modified system, stochastically dominates the queue lengths process in the original system (see \cite{szpa,szpa1}). Therefore, proving stability of such
a dominant system - that is, the one under the partition $(S,U)$ - suffices for stability of
the original system. To accomplish this, we prove stability conditions for users in $S$ by mathematical induction. Finally, the stability region for the original system is a union of stability regions obtained for every partition $P$; see Theorem $1$ in \cite{szpa1}. However, as proved in \cite{szpa1}, only partitions $P_{k}=(M_{k},\{k\})$, where $M_{k}=M-\{k\}$ contribute to the final stability region.
\subsection{An alternative approach for the case of $N=2$ users}
We will present an alternative approach for the derivation of the stability region for the case of $N=2$ users in order to assist the analysis for the case of $N=3$ users. For such a case we consider the partitions $P_{1}=(M_{1},\{1\})$, $P_{2}=(M_{2},\{2\})$, where $M_{1}=\{2\}$ and $M_{2}=\{1\}$, and let $\mathcal{R}_{i}$ be stability region for the partition $P_{i}$, $i=1,2$. We will discuss the construction of $\mathcal{R}_{1}$ in detail, then similarly the rest can be obtained. Denote by $P_{suc}^{(i)}(M_{j})$ the probability of a successful transmission from user $i$ in the dominant system $M_{j}$. Clearly,
\begin{displaymath}
\begin{array}{rl}
P_{suc}^{(1)}(M_{1})=&\alpha_{1}^{*}\tilde{P}_{1/\{1\}}Pr(N_{2}=0)+\alpha_{1}(\bar{\alpha}_{2}P_{1/\{1\}}+\alpha_{2}P_{1/\{1,2\}})Pr(N_{2}>0)
=\alpha_{1}^{*}\tilde{P}_{1/\{1\}}+d_{1}P(N_{2}>0),\\
P_{suc}^{(2)}(M_{1})=&\alpha_{2}(\bar{\alpha}_{1}P_{2/{2}}+\alpha_{1}P_{2/\{1,2\}})=\alpha_{2}\widehat{\alpha}_{1}.
\end{array}
\end{displaymath}
However, for $\lambda_{2}<P_{suc}^{(2)}(M_{1})$, $Pr(N_{2}>0)=\lambda_{2}/(\alpha_{2} \widehat{\alpha}_{1})$, and hence $\mathcal{R}_{1}$ is obtained. Similarly, by considering $M_{2}$ we can obtain the stability region $\mathcal{R}_{2}$. Thus, stability conditions are the same with ones obtained in Theorem \ref{lem}.

\subsection{The case of $N=3$ users}
Here we consider the case of $N=3$ users, which is more intricate. We now have to investigate only three partitions $P_{i}=(M_{i},\{i\})$, where $M_{1}=\{2,3\}$, $M_{2}=\{1,3\}$ and $M_{3}=\{1,2\}$, and only the first partition will be discussed in detail. As stated previously, the stability region $\mathcal{R}$ is the union of three regions $\mathcal{R}_{1}$, $\mathcal{R}_2$ and $\mathcal{R}_3$ each corresponding to $M_1$, $M_2$, and $M_3$, respectively. 

To proceed, we have to make clear how the system operates: for convenience we assume that node $i$, $i=1,2,3$, transmits a packet to the common destination with probability $\alpha_{i}$ when the node $(i\text{ }mod\ 3 +1)$ is non empty, and with probability $\alpha_{i}^{*} \geq \alpha_{i}$ when the node $(i\text{ }mod\ 3+1)$ is empty. We  will present the derivation of $\mathcal{R}_{1}$. In the corresponding dominant system, the first user never empties. Note that such a system can be viewed as a two-user system with an additional user who creates interference (i.e., it transmits dummy packets when is is empty) to the other users. In order to proceed we are going to perform a similar analysis as in Section \ref{sec:bound}.

Let $F_{1}(y,z)$ be the generating function of $(N_{2,n},N_{3,n})$  with the first user being an interfering one (i.e., it never empties). Then, with a minor modification, 
\begin{equation}
\begin{array}{rl}
\lambda_{2}=&\alpha_{2}\widehat{\alpha}_{13}^{(2)}(1-F_{1}(0,1))-d_{2}^{*}(F_{1}(1,0)-F_{1}(0,0)),\\
\lambda_{3}=&\alpha_{3}\widehat{\alpha}_{12}^{(3)}(1-F_{1}(1,0))-d_{3}^{*}(F_{1}(0,1)-F_{1}(0,0)),
\end{array}
\label{as}
\end{equation}
where, for $k\neq i\neq j$, $k,i,j\in M=\{1,2,3\}$,
\begin{displaymath}
\begin{array}{rl}
\widehat{\alpha}_{ij}^{(k)}=&\bar{\alpha}_{i}\bar{\alpha}_{j}P_{k/\{k\}}+\alpha_{i}\bar{\alpha}_{j}P_{k/\{k,i\}}+\alpha_{j}\bar{\alpha}_{i}P_{k/\{k,j\}}+\alpha_{i}\alpha_{j}P_{k/\{k,i,j\}},\\
d_{2}^{*}=&\alpha_{2}\widehat{\alpha}_{13}^{(2)}-\alpha_{2}^{*}\bar{\alpha}_{12},\,
d_{3}^{*}=\alpha_{3}(\widehat{\alpha}_{12}^{(3)}-\bar{\alpha}_{13}^{*}),\\
\bar{\alpha}_{ij}=&\bar{\alpha}_{i}\tilde{P}_{j/\{j\}}+\alpha_{i}\tilde{P}_{j/\{i,j\}},\,i,j\in M,\,i\neq j\\
\bar{\alpha}_{ij}^{*}=&\bar{\alpha}_{i}^{*}\tilde{P}_{j/\{j\}}+\alpha_{i}^{*}\tilde{P}_{j/\{i,j\}},\,i,j\in M,\,i\neq j,
\end{array}
\end{displaymath}
and $P_{i/T}$ is the success probability of user $i$ when the transmitting users are in $T$ and all users have packets to send, $\tilde{P}_{i/\{T\}}$ is the success probability of user $i$ when the transmitting users are in $T$ and there is only one user ($\neq i,1$) that is empty.
Following \cite{szpa1}, for $z_{2},z_{3}\in\{0,1\}$, let $P_{1}(z_{2},z_{3})=P(\chi(\bar{N_{2}})=z_{2},\chi(\bar{N_{3}})=z_{3})$, with the first user be the interfering one, $(\bar{N}_{2},\bar{N}_{3})$ the queue lengths in the modified system and $\chi(k)=0$ for $k=0$ and $\chi(k)=1$ for $k\geq 1$. Then, from the analysis in section \ref{sec:bound}, we have
\begin{equation}
\begin{array}{c}
P_{1}(0,0)=\frac{\alpha_{2}\widehat{\alpha}_{13}^{(2)}(\alpha_{3}\widehat{\alpha}_{12}^{(3)}-\lambda_{3})+d_{3}^{*}(\lambda_{2}-\alpha_{2}\widehat{\alpha}_{13}^{(2)})}{\alpha_{2}\widehat{\alpha}_{13}^{(2)}\alpha_{3}\bar{\alpha}_{13}^{*}}\exp \left[\frac{-\gamma(1)}{2\pi i}\int_{|t|=1}\frac{\log\{J(t)\}}{t(t-\gamma(1))}dt \right],
\end{array}
\label{es}
\end{equation}
where $\gamma(x)$ is the inverse of a conformal mapping of the unit circle onto a curve $\mathcal{M^{*}}$ (see subsection \ref{ker}). Note also that $P_{1}(1,0)=F_{1}(1,0)-F_{1}(0,0)$, $P_{1}(0,1)=F_{1}(0,1)-F_{1}(0,0)$, $P_{1}(1,1)=1-F_{1}(1,0)-F_{1}(0,1)+F_{1}(0,0)$, $P_{1}(0,0)=F_{1}(0,0)$. Now, for $P_{1}=(M_{1},\{1\})$, and after some algebra,
\begin{equation}
\begin{array}{rl}
P_{suc}^{(1)}(M_{1})=&\alpha_{1}^{*}\tilde{P}_{1/\{1\}}^{\prime}P_{1}(0,0)+\alpha_{1}\bar{\alpha}_{21}^{*}P_{1}(1,0)+\alpha_{1}^{*}\bar{\alpha}_{31}P_{1}(0,1)+\alpha_{1}\widehat{\alpha}_{23}^{(1)}P_{1}(1,1),\\
P_{suc}^{(2)}(M_{1})=&\alpha_{2}\widehat{\alpha}_{13}^{(2)}+F_{1}(1,0)(\alpha_{2}^{*}\bar{\alpha}_{12}-\alpha_{2}\widehat{\alpha}_{13}^{(2)}),\\
P_{suc}^{(3)}(M_{1})=&\alpha_{3}\widehat{\alpha}_{12}^{(3)}+F_{1}(0,1)\alpha_{3}(\bar{\alpha}_{13}^{*}-\widehat{\alpha}_{12}^{(3)}),
\end{array}
\label{wsx}
\end{equation}
where $\tilde{P}_{i/\{i\}}^{\prime}$ is the success probability for a user $i$, when it is the only non-empty node.

Note that from (\ref{as}), provided that $d_{2}^{*}d_{3}^{*}-\alpha_{2}\widehat{\alpha}_{13}^{(2)}\alpha_{3}\widehat{\alpha}_{12}^{(3)}\neq 0$,
\begin{equation}
\begin{array}{rl}
F_{1}(1,0)=&\frac{\alpha_{2}\widehat{\alpha}_{13}^{(2)}(\alpha_{3}\widehat{\alpha}_{12}^{(3)}-\lambda_{3})+d_{3}^{*}(\lambda_{2}-\alpha_{2}\widehat{\alpha}_{13}^{(2)})+d_{3}^{*}\alpha_{2}^{*}\bar{\alpha}_{12}F_{1}(0,0)}{\alpha_{2}\widehat{\alpha}_{13}^{(2)}\alpha_{3}\widehat{\alpha}_{12}^{(3)}-d_{2}^{*}d_{3}^{*}},\\
F_{1}(0,1)=&\frac{\alpha_{3}\widehat{\alpha}_{12}^{(3)}(\alpha_{2}\widehat{\alpha}_{13}^{(2)}-\lambda_{2})+d_{2}^{*}(\lambda_{3}-\alpha_{3}\widehat{\alpha}_{12}^{(3)})+d_{2}^{*}\alpha_{3}\bar{\alpha}_{13}^{*}F_{1}(0,0)}{\alpha_{2}\widehat{\alpha}_{13}^{(2)}\alpha_{3}\widehat{\alpha}_{12}^{(3)}-d_{2}^{*}d_{3}^{*}}.
\end{array}
\label{cx}
\end{equation}
Therefore, after some simple but tedious calculations we have
\begin{equation}
\begin{array}{rl}
P_{suc}^{(1)}(M_{1})=&P_{1}(0,0)\{\alpha_{1}^{*}\tilde{P}_{1/\{1\}}^{\prime}-\alpha_{1}\bar{\alpha}_{21}^{*}-\alpha_{1}^{*}\bar{\alpha}_{31}-\alpha_{1}\widehat{\alpha}_{23}^{(1)}+\frac{(\alpha_{1}\bar{\alpha}_{21}^{*}-\alpha_{1}\widehat{\alpha}_{23}^{(1)})d_{3}^{*}\alpha_{2}^{*}\bar{\alpha}_{12}+(\alpha_{1}^{*}\bar{\alpha}_{31}-\alpha_{1}\widehat{\alpha}_{23}^{(1)})d_{2}^{*}\alpha_{3}\bar{\alpha}_{13}^{*}}{\alpha_{2}\widehat{\alpha}_{13}^{(2)}\alpha_{3}\widehat{\alpha}_{12}^{(3)}-d_{2}^{*}d_{3}^{*}}\}\\
&+\frac{(\alpha_{3}\widehat{\alpha}_{12}^{(3)}-\lambda_{3})[\alpha_{2}\widehat{\alpha}_{13}^{(2)}(\alpha_{1}\bar{\alpha}_{21}^{*}-\alpha_{1}\widehat{\alpha}_{23}^{(1)})-d_{2}^{*}(\alpha_{1}^{*}\bar{\alpha}_{31}-\alpha_{1}\widehat{\alpha}_{23}^{(1)})]}{\alpha_{2}\widehat{\alpha}_{13}^{(2)}\alpha_{3}\widehat{\alpha}_{12}^{(3)}-d_{2}^{*}d_{3}^{*}}
+\frac{(\alpha_{2}\widehat{\alpha}_{13}^{(2)}-\lambda_{2})[\alpha_{3}\widehat{\alpha}_{12}^{(3)}(\alpha_{1}^{*}\bar{\alpha}_{31}-\alpha_{1}\widehat{\alpha}_{23}^{(1)})-d_{3}^{*}(\alpha_{1}\bar{\alpha}_{21}^{*}-\alpha_{1}\widehat{\alpha}_{23}^{(1)})]}{\alpha_{2}\widehat{\alpha}_{13}^{(2)}\alpha_{3}\widehat{\alpha}_{12}^{(3)}-d_{2}^{*}d_{3}^{*}}.
\end{array}
\label{es1}
\end{equation}
Similarly, using (\ref{cx}) we can express $P_{suc}^{(2)}(M_{1})$, $P_{suc}^{(3)}(M_{1})$ in terms of $F_{1}(0,0)$.

In summary, following \cite{szpa1} we have the following corollary.
\begin{corollary}
The system with $N=3$ users is stable if and only if $(\lambda_{1},\lambda_{2},\lambda_{3})\in \mathcal{R}=\mathcal{R}_1 \cup \mathcal{R}_2 \cup \mathcal{R}_3$, where
\begin{displaymath}
\scriptsize
\begin{array}{rl}
\mathcal{R}_{1}=\left\{(\lambda_{1},\lambda_{2}) :\lambda_{1}<P_{suc}^{(1)}(M_{1}),\,\lambda_{2}<\alpha_{2}\widehat{\alpha}_{13}^{(2)}+F_{1}(1,0)(\alpha_{2}^{*}\bar{\alpha}_{12}-\alpha_{2}\widehat{\alpha}_{13}^{(2)}),\lambda_{3}<\alpha_{3}\widehat{\alpha}_{12}^{(3)}+F_{1}(0,1)\alpha_{3}(\bar{\alpha}_{13}^{*}-\widehat{\alpha}_{12}^{(3)})\right\},\\
\mathcal{R}_{2}=\left\{(\lambda_{1},\lambda_{2}) :\lambda_{1}<\alpha_{1}\widehat{\alpha}_{23}^{(1)}+F_{2}(1,0)\alpha_{1}(\bar{\alpha}_{21}^{*}-\widehat{\alpha}_{23}^{(1)}),\,\lambda_{2}<P_{suc}^{(2)}(M_{2}),\,\lambda_{3}<\alpha_{3}\widehat{\alpha}_{12}^{(3)}+F_{2}(0,1)(\alpha_{3}^{*}\bar{\alpha}_{23}-\alpha_{3}\widehat{\alpha}_{12}^{(3)})\right\},\\

\mathcal{R}_{3}=\left\{(\lambda_{1},\lambda_{2}) :\lambda_{1}<\alpha_{1}\widehat{\alpha}_{23}^{(1)}+F_{3}(1,0)(\alpha_{1}^{*}\bar{\alpha}_{31}-\alpha_{1}\widehat{\alpha}_{23}^{(1)}),\,\lambda_{2}<\alpha_{2}\widehat{\alpha}_{13}^{(2)}+F_{3}(0,1)\alpha_{2}(\bar{\alpha}_{32}^{*}-\widehat{\alpha}_{13}^{(2)}),\,\lambda_{3}<P_{suc}^{(3)}(M_{3})\right\},
\end{array}
\end{displaymath}
where the appropriate probabilities are computed from the results obtained in Section \ref{sec:bound} as discussed in (\ref{as})-(\ref{es1}). 
\end{corollary}

\section{Explicit expressions for the symmetrical model for the Two-user case}
\label{sec:symmetric2user}
In this section we consider the symmetrical model and obtain closed form expressions for the average delay for the collision model and the capture model without explicitly computing the generating function for the stationary joint queue length distribution. Moreover, we provide upper and lower delay bounds for the MPR channel model.

By symmetrical, we mean the case where $\alpha_{i}^{*}=\alpha^{*}$, $\alpha_{i}=\alpha$, $\lambda_{i}=\lambda$, $P_{i/\{i\}}=p$, $\tilde{P}_{i/\{i\}}=\tilde{p}$, $P_{i/\{1,2\}}=b$, $P_{1,2/\{1,2\}}=c$. Due to the symmetry of the model we have $H_{1}(1,1)=H_{2}(1,1)$, $H_{1}(1,0)=H_{2}(0,1)$. Note that $M_{k}=H_{k}(1,1)$ the expected number of packets in node $j$. Therefore, after simple calculations using (\ref{we}) we obtain,
\begin{equation}
M_{1}=\frac{\lambda+(d+\alpha^{2}c)H_{1}(1,0)}{\alpha(p+\alpha(b+c-p))-\lambda}.
\label{t1}
\end{equation}
Setting $x=y$ in (\ref{we}), differentiating it with respect to $x$ at $x=1$, and using (\ref{r1}) we obtain
\begin{equation}
\begin{array}{c}
M_{1}+M_{2}=2M_{1}=\frac{2\lambda-\lambda^{2}+\alpha^{2}cP(N_{1}>0,N_{2}>0)+2H_{1}(1,0)(\alpha(p+\alpha(b-p)+d+2\alpha^{2}c))}{2(\alpha(p+\alpha(b+c-p))-\lambda)}.
\end{array}
\label{t2}
\end{equation}
Using (\ref{t1}), (\ref{t2}) we finally obtain
\begin{equation}
\begin{array}{c}
M_{1}=M_{2}=\frac{\lambda[2(\alpha+\alpha^{2}(b+c-p))+\lambda(d+\alpha^{2}c)]}{2\alpha^{*}\tilde{p}(\alpha(p+\alpha(b+c-p))-\lambda)}-\frac{\alpha^{2}c(d+\alpha^{2}c)P(N_{1}>0,N_{2}>0)}{2\alpha^{*}\tilde{p}(\alpha(p+\alpha(b+c-p))-\lambda)}.
\end{array}
\label{rt}
\end{equation}
Therefore, using Little's law the average delay in a node is given by
\begin{equation}
\begin{array}{c}
D_{1}=D_{2}=\frac{2(\alpha+\alpha^{2}(b+c-p))+\lambda(d+\alpha^{2}c)}{2\alpha^{*}\tilde{p}(\alpha(p+\alpha(b+c-p))-\lambda)}+\phi,
\end{array}
\label{rt1}
\end{equation}
where $\phi=-\frac{\alpha^{2}c(d+\alpha^{2}c)P(N_{1}>0,N_{2}>0)}{2\lambda\alpha^{*}\tilde{p}(\alpha(p+\alpha(b+c-p))-\lambda)}$;
note that $\alpha(p+\alpha(b+c-p))>\lambda$ due to the stability condition. 

In case of the capture model, i.e., $c=0$, the exact average queueing delay in a node is given by (\ref{rt1}) for $\phi=0$. In case $c\neq0$, i.e., strong MPR effect, we are going to obtain upper and lower bounds for the expected delay based on the sign of $\phi$. Since $P(N_{1}>0,N_{2}>0)>0$, the sign of $\phi$ coincides with the sign of $d+\alpha^{2}c$. Thus, in order to proceed, we distinguish the analysis in the following two cases:
\begin{enumerate}
\item If $d+\alpha^{2}c<0$, then $0\leq\phi\leq -\frac{\alpha^{2}c(d+\alpha^{2}c)}{2\lambda\alpha^{*}\tilde{p}(\alpha(p+\alpha(b+c-p))-\lambda)}$.
Thus, the upper and lower delay bound, say $D_{1}^{up}$, $D_{1}^{low}$ respectively are,
\begin{displaymath}
\begin{array}{rl}
D_{1}^{up}=&\frac{2(\alpha+\alpha^{2}(b+c-p))+\lambda(d+\alpha^{2}c)}{2\alpha^{*}\tilde{p}(\alpha(p+\alpha(b+c-p))-\lambda)}-\frac{\alpha^{2}c(d+\alpha^{2}c)}{2\lambda\alpha^{*}\tilde{p}(\alpha(p+\alpha(b+c-p))-\lambda)},\\
D_{1}^{low}=&\frac{2(\alpha+\alpha^{2}(b+c-p))+\lambda(d+\alpha^{2}c)}{2\alpha^{*}\tilde{p}(\alpha(p+\alpha(b+c-p))-\lambda)}.
\end{array}
\end{displaymath} 
\item If $d+\alpha^{2}c>0$, then $-\frac{\alpha^{2}c(d+\alpha^{2}c)}{2\widehat{\lambda}\alpha^{*}\tilde{p}(\alpha(p+\alpha(b+c-p))-\widehat{\lambda})}\leq\phi\leq 0$. 
In such a case,
\begin{displaymath}
\begin{array}{rl}
D_{1}^{up}=&\frac{2(\alpha+\alpha^{2}(b+c-p))+\lambda(d+\alpha^{2}c)}{2\alpha^{*}\tilde{p}(\alpha(p+\alpha(b+c-p))-\lambda)},\\
D_{1}^{low}=&\frac{2(\alpha+\alpha^{2}(b+c-p))+\lambda(d+\alpha^{2}c)}{2\alpha^{*}\tilde{p}(\alpha(p+\alpha(b+c-p))-\lambda)}-\frac{\alpha^{2}c(d+\alpha^{2}c)}{2\lambda\alpha^{*}\tilde{p}(\alpha(p+\alpha(b+c-p))-\lambda)}.
\end{array}
\end{displaymath} 
\end{enumerate}

\begin{remark}
Note that $d+\alpha^{2}c=\alpha p+\alpha^{2}(b+c-p)-\alpha^{*}\tilde{p}$, is the difference of the successful transmission probability of a node when both nodes are active (i.e., $\alpha p+\alpha^{2}(b+c-p)$) minus the successful transmission probability of a node when the other node is inactive (i.e., $\alpha^{*}\tilde{p}$). Clearly, it is realistic to assume that $d+\alpha^{2}c<0$, since it is more likely for a node to successfully transmit a packet when it is the only active. 
\end{remark}

\begin{lemma}
Let $\tilde{\alpha}$ be the optimal transmission probability for the minimizing the expected delay in the symmetric capture model with $\alpha\leq\alpha^{*}\leq1$. Then,
\begin{displaymath}
\begin{array}{c}
\tilde{\alpha}=\left\{\begin{array}{ll}
\alpha^{*},&\,if\,b\geq\frac{p(2\alpha^{*}-1)}{2\alpha^{*}},\vspace{2mm}\\
\frac{p}{2(p-b)},&if\,0\leq b<\frac{p(2\alpha^{*}-1)}{2\alpha^{*}}.
\end{array}\right.
\end{array}
\end{displaymath}
\end{lemma}
\begin{proof}
The problem can be cast as follows:
\begin{equation}
\begin{array}{c}
\tilde{\alpha}=\underset{\{\alpha(p+\alpha(b-p))-\lambda>0,\,\alpha\in[0,\alpha^{*}]\}}{argmin}
\left\{\frac{(2+\lambda)(\alpha+\alpha^{2}(b-p))+\lambda\alpha^{*}\tilde{p}}{2\alpha^{*}\tilde{p}(\alpha(p+\alpha(b-p))-\lambda)}\right\}.
\end{array}
\label{dc0}
\end{equation}
To proceed, we first focus on the looser constrained optimization problem,
\begin{equation}
\begin{array}{c}
\alpha^{\prime}=\underset{\{\alpha(p+\alpha(b-p))-\lambda>0\}}{argmin}
\left\{\frac{(2+\lambda)(\alpha+\alpha^{2}(b-p))+\lambda\alpha^{*}\tilde{p}}{2\alpha^{*}\tilde{p}(\alpha(p+\alpha(b-p))-\lambda)}\right\}.
\end{array}
\label{dc}
\end{equation} 
Clearly $b<p$, since it is more likely a transmission to be successful when only one node is transmitting rather than when both nodes transmit. Thus, $\alpha(p+\alpha(b-p))-\lambda>0$ is equivalent with $s_{1}<\alpha<s_{2}$, where $s_{1}$, $s_{2}$ the roots of $\alpha(p+\alpha(b-p))-\lambda=0$, where $0\leq s_{1}\leq1$. Differentiating the objective function in (\ref{dc}), we can easily derive that the only possible minimum will be given at $\alpha^{\prime}=\frac{p}{2(p-b)}$, where $s_{1}\leq \alpha^{\prime}\leq s_{2}$. If $\alpha^{\prime}<\alpha^{*}$, which is true for $b<\frac{p(2\alpha^{*}-1)}{2\alpha^{*}}$, then $\tilde{\alpha}=\alpha^{\prime}=\frac{p}{2(p-b)}$ is the minimum of the objective function (\ref{dc0}). If $\alpha^{\prime}\geq\alpha^{*}$, which is equivalent with $b\geq\frac{p(2\alpha^{*}-1)}{2\alpha^{*}}$, then the optimal transmission probability, which minimizes the objective function in (\ref{dc0}) is $\tilde{\alpha}=\alpha^{*}$.
\end{proof}

\section{Numerical results}\label{sec:num}
In this section, we provide numerical results to validate the analysis presented earlier. We consider the case where the users have the same link characteristics and transmission probabilities to facilitate exposition clarity, so we will use the notation from Section \ref{sec:symmetric2user}.

\subsection{Stable Throughput Region}
The stability or stable throughput region for given transmission probabilities is depicted in Fig. \ref{fig:region} in the general case. The proposed random access scheme for given transmission probabilities is superior in the cases of collision, capture and the MPR channel modes, as it can be easily seen by replacing the parameters and putting $\alpha_{1}^{*}=\alpha_{1}$ and $\alpha_{2}^{*}=\alpha_{2}$. 

As mentioned above, in Section \ref{sec:stability2}, we obtained the stability region with fixed transmission probability vectors $(\alpha_{1}, \alpha_{2},\alpha_{1}^{*}, \alpha_{2}^{*})$.
If we take the union of these regions over all possible transmission probabilities of the users, we obtain the total stability region (i.e. the envelope of the individual regions). This corresponds to the closure of the stability region and is defined as
\begin{equation} \label{eq:closure_def}
\mathcal{L}\triangleq \left( \bigcup_{ \boldsymbol{\vec{\alpha}} \in [0,1]^2 \times [\alpha_{1},1] \times [\alpha_{2},1]} \mathcal{L}_1 (\boldsymbol{\vec{\alpha}}) \right) \bigcup \left( \bigcup_{ \boldsymbol{\vec{\alpha}} \in [0,1]^2 \times [\alpha_{1},1] \times [\alpha_{2},1]} \mathcal{L}_2 (\boldsymbol{\vec{\alpha}}) \right)
\end{equation}
where $\mathcal{L}_i (\boldsymbol{\vec{\alpha}} ) \triangleq \mathcal{R}_i$ for $i=1,2$ are obtained in Section \ref{sec:stability2} and $\boldsymbol{\vec{\alpha}}=(\alpha_{1}, \alpha_{2},\alpha_{1}^{*}, \alpha_{2}^{*})$ is the vector of transmission probabilities.

Here, we will present the closure of the stable throughput region for the collision channel case where  $p=\tilde{p}=1$ and $b=c=0$.\footnote{The closure for the MPR channel model is omitted due to space limitations.} In Figs. \ref{fig:STR_collision} and \ref{fig:STR_collision2} the closure of the stability region for the traditional collision channel with random access and for the proposed scheme are depicted. Clearly, our scheme is superior to the traditional one. The region in Fig. \ref{fig:STR_collision2} is broader than the one in Fig. \ref{fig:STR_collision} which means that higher arrival rates can be supported and still maintain the system stable. Besides, the shape of the closure of the proposed scheme has linear behavior compare to the non-linear for the traditional one. This is a very interesting result. 

\begin{figure}[h!]
\centering
\subfigure[Collision channel with $\alpha=\alpha^{*}$.]{
\includegraphics[scale=0.45]{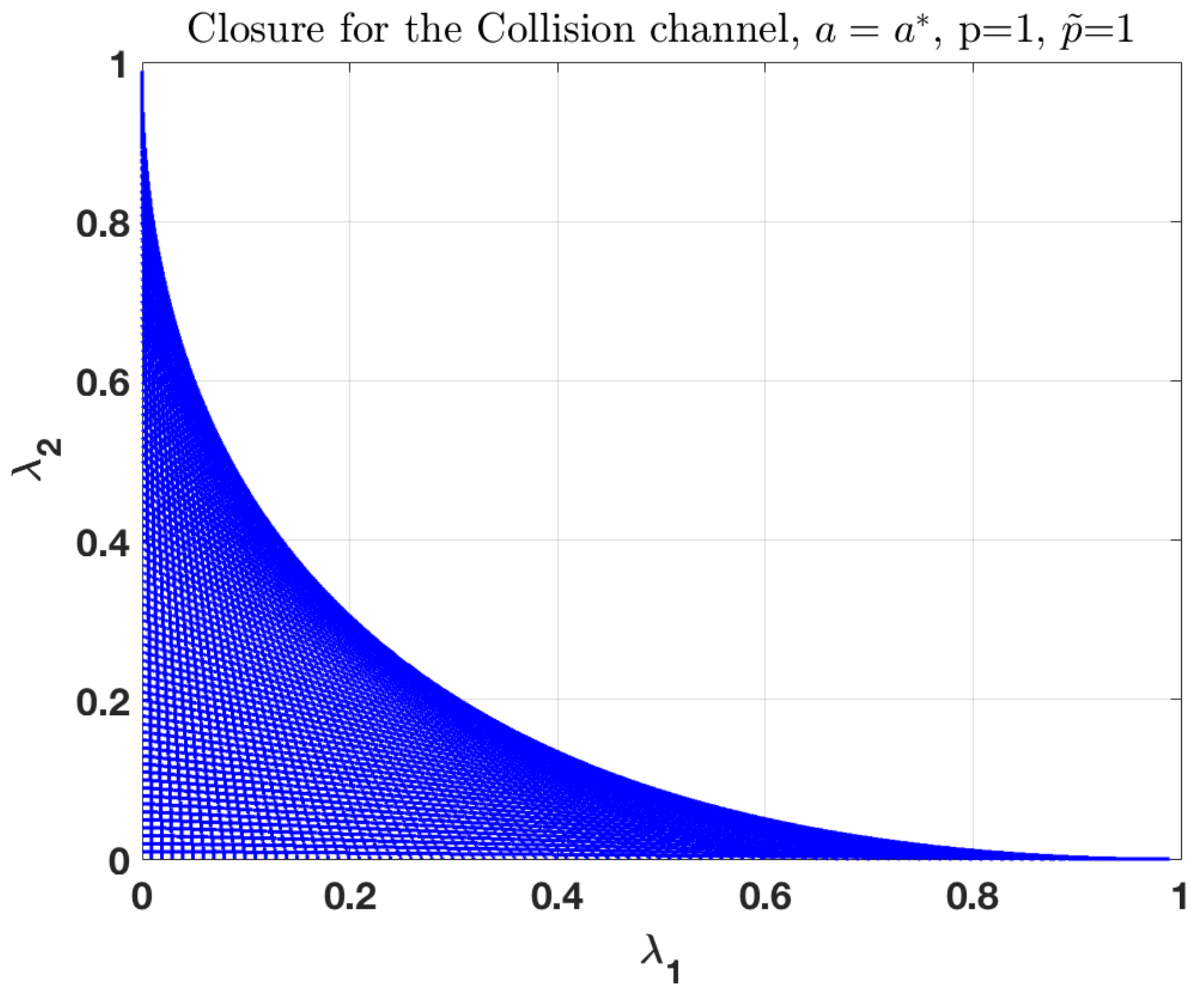}
\label{fig:STR_collision}
}
\subfigure[Collision channel with $\alpha \leq \alpha^{*}\leq1$.]{
\includegraphics[scale=0.45]{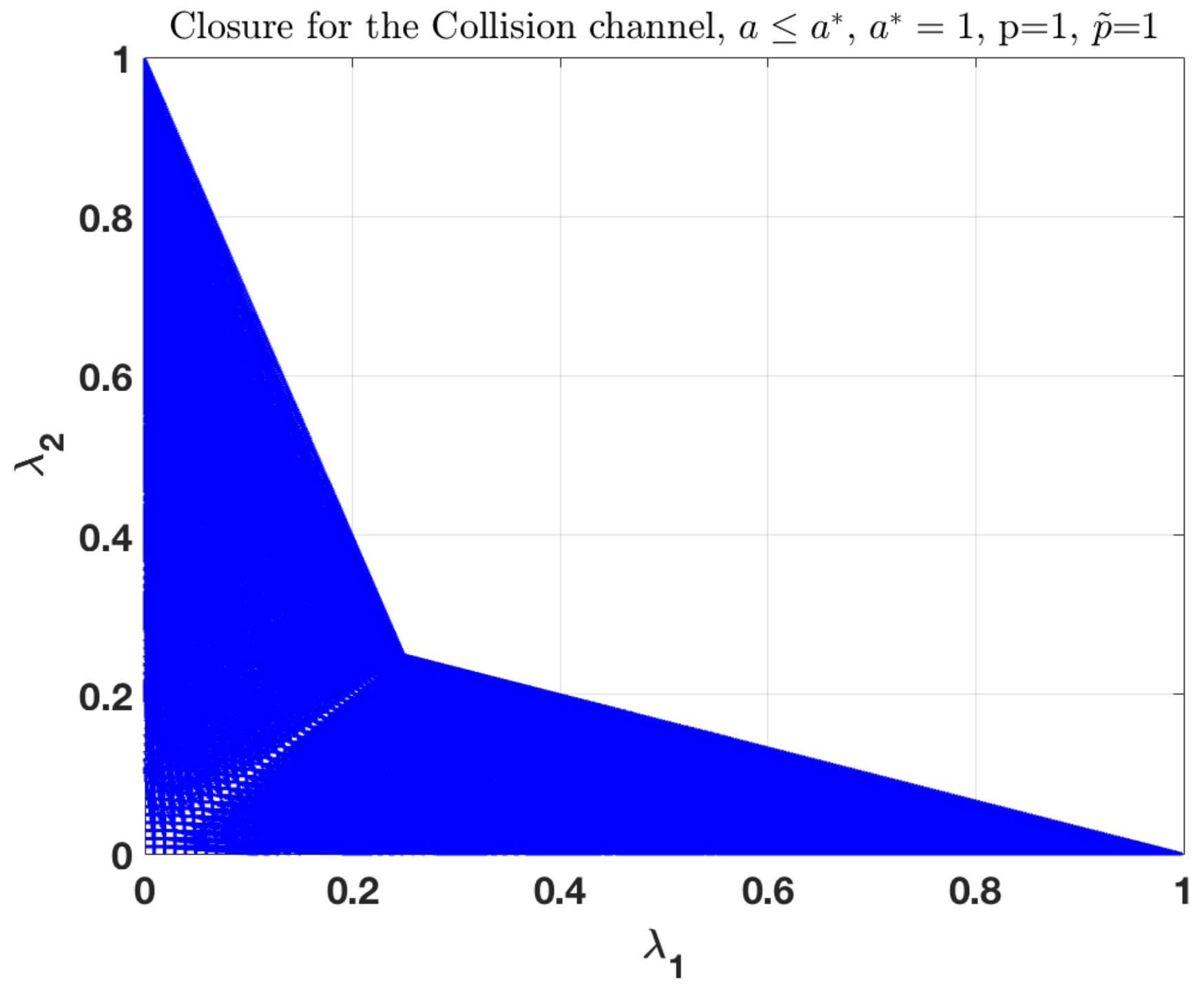}
\label{fig:STR_collision2}
}
\caption{Closure of the Stability Region for the collision channel ($b=c=0$) for $p=\tilde{p}=1$.}
\end{figure}

In Figs. \ref{fig:STR_capture} and \ref{fig:STR_capture2} the closure of the stability region for the capture channel with random access and for the proposed scheme are depicted for $b=0.2$. Our scheme is still superior to the traditional one since the region in Fig. \ref{fig:STR_capture} is a subset of the region in Fig. \ref{fig:STR_capture}.

\begin{figure}[h!]
\centering
\subfigure[Capture channel with $\alpha=\alpha^{*}$.]{
\includegraphics[scale=0.45]{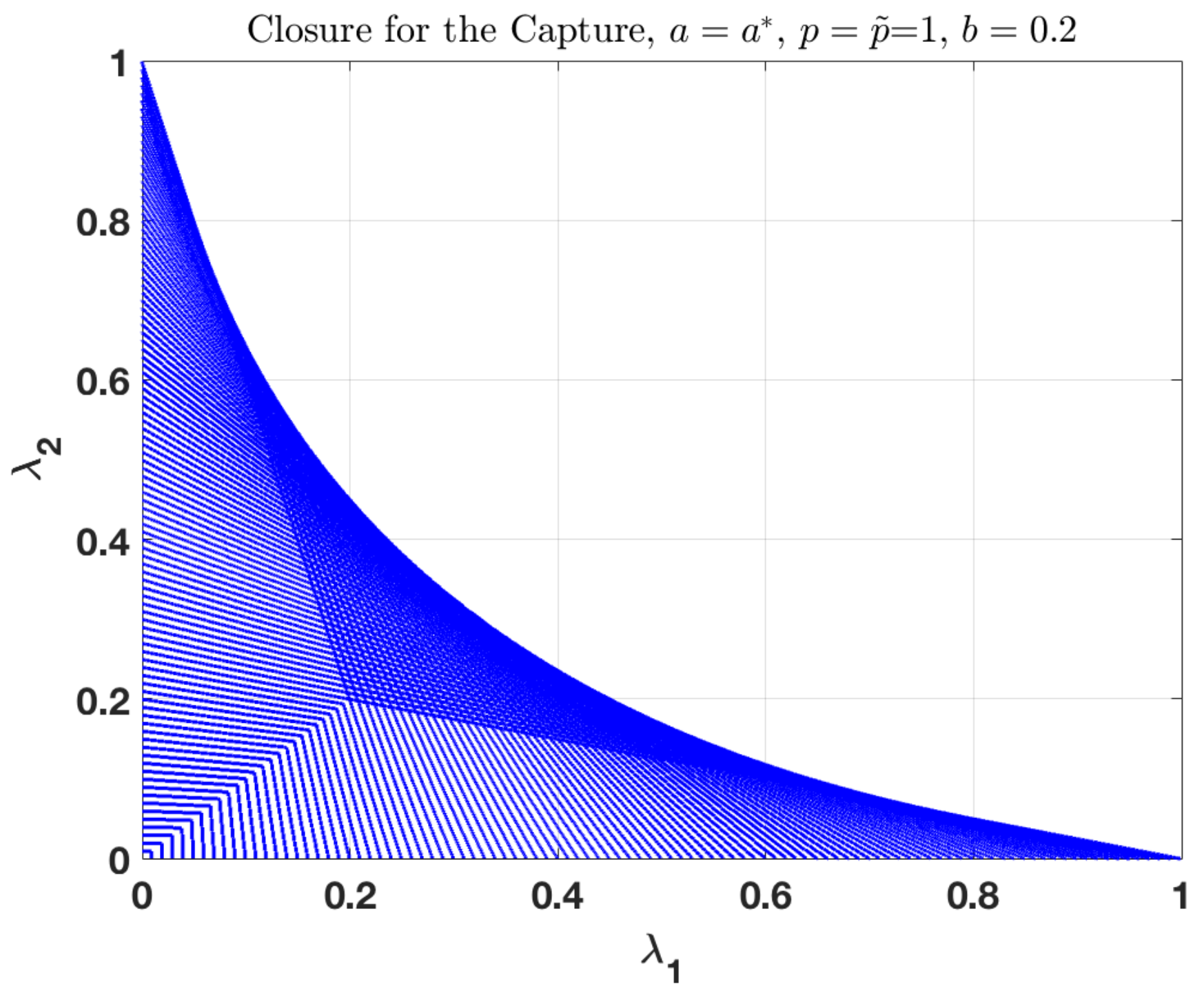}
\label{fig:STR_capture}
}
\subfigure[Capture channel with $\alpha \leq \alpha^{*}\leq1$.]{
\includegraphics[scale=0.45]{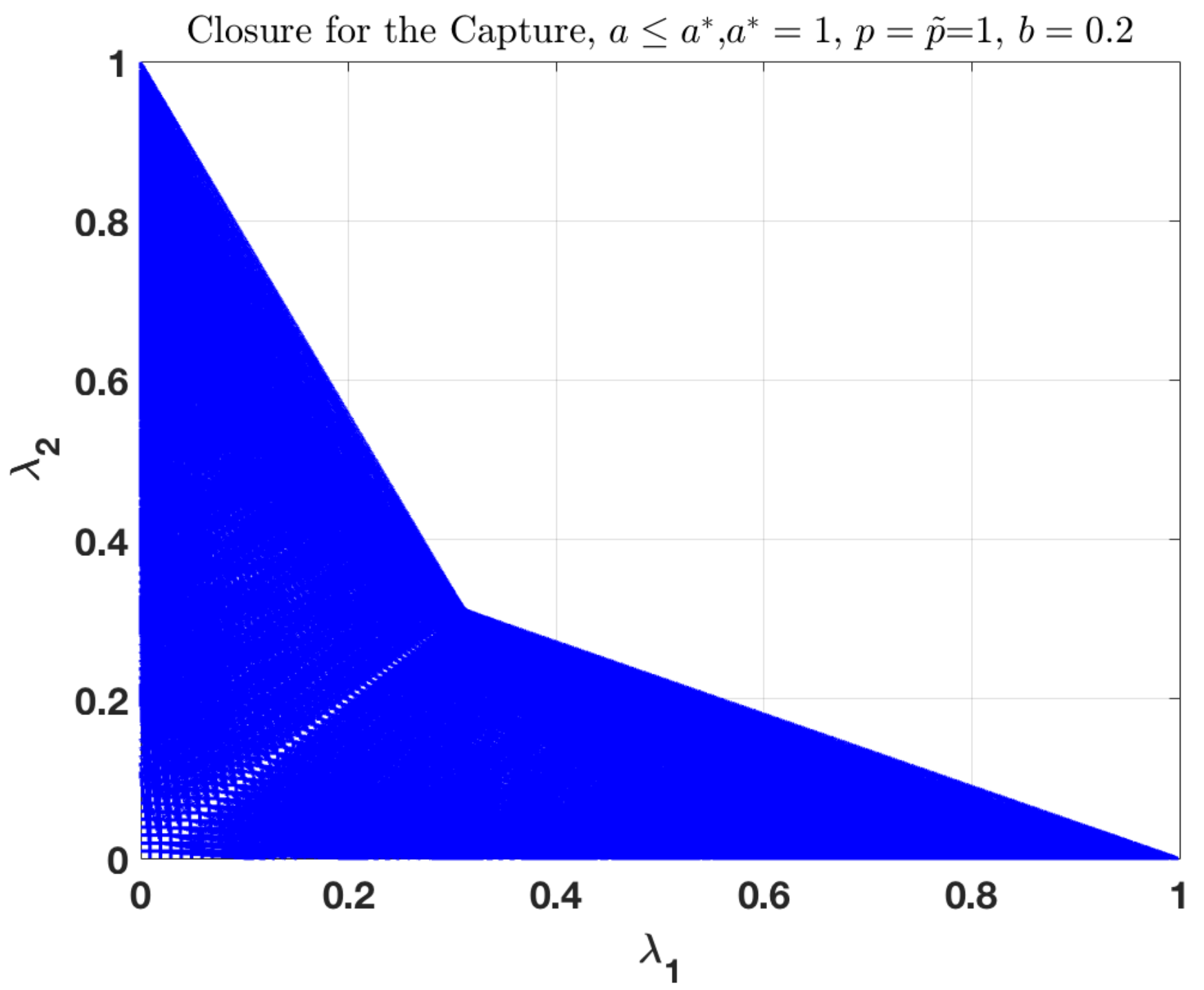}
\label{fig:STR_capture2}
}
\caption{Closure of the Stability Region for the capture channel ($b=0.2, c=0$) for $p=\tilde{p}=1$.}
\end{figure}

\subsection{Average Delay}
The effect of the arrival rate $\lambda$ at the average delay is depicted in Fig. \ref{delay_vs_lambda} for the collision, capture and the MPR channel models.
We consider the case with $\alpha=0.6$, $\alpha^{*}=1$ and $p=0.9$, $\tilde{p}=1$.
Clearly, regarding the MPR channel model, the lower and the upper bounds appear to be close. As also expected the average delay is lower for the MPR than the capture and the collision. As also expected finite delay can be sustained for larger values of $\lambda$ for the MPR case.

\begin{figure}[h!]
\centering
\subfigure[Effect of $\lambda$ on the average delay.]{
\includegraphics[scale=0.45]{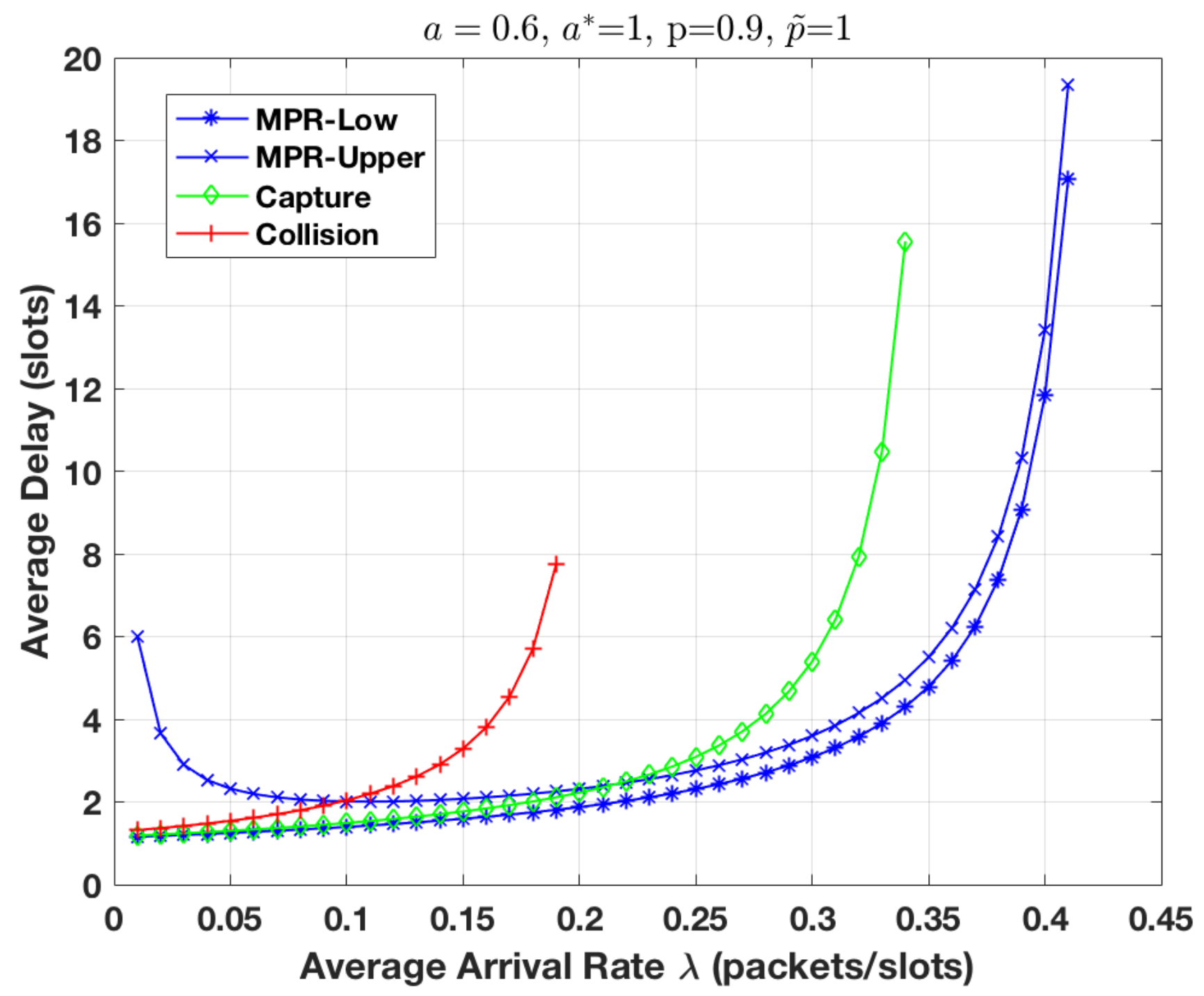}
\label{delay_vs_lambda}
}
\subfigure[Effect of $\alpha^{*}$ as $\lambda$ varies.]{
\includegraphics[scale=0.45]{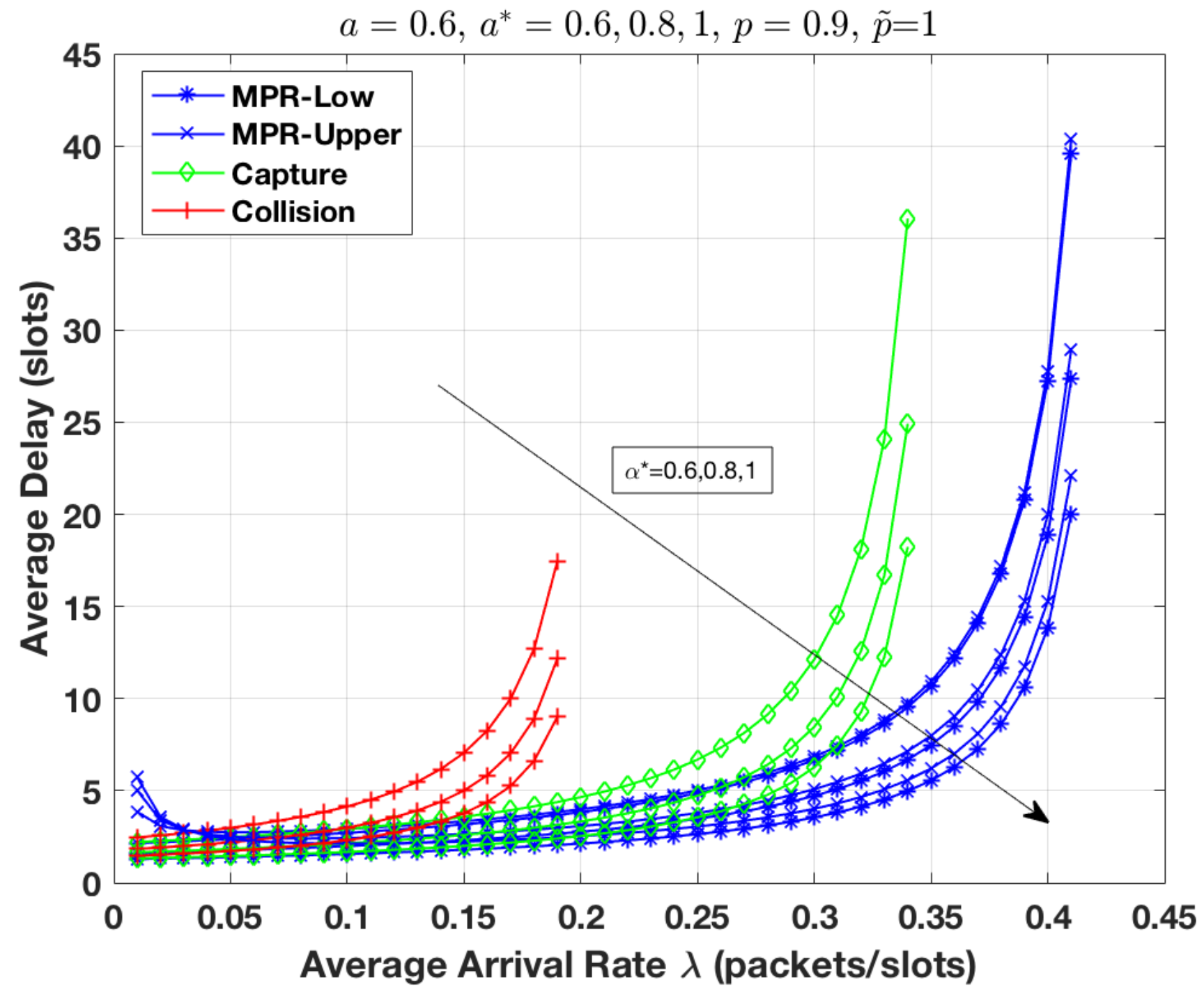}
\label{delay_vs_alphastar}
}
\caption{Effect of $\lambda$ on the average delay for the collision, capture and the MPR channel models.}
\end{figure}

In Fig. \ref{delay_vs_alphastar} we present the effect of $\alpha^{*}$ on the average delay as $\lambda$ varies. The cases of the collision, capture and the MPR channel models are presented. As $\alpha^{*}$ increases then average delay decreases and also the maximum arrival rate that can still maintain a finite delay is getting larger. Clearly, adapting the transmission probabilities depending on the queue state can increase the performance of the system.

\section{Conclusions}
In this work we considered the case of the two and three-user with bursty traffic in a random access wireless network with a common destination that employs MPR capabilities. We assumed that the users adapt their transmission probabilities based on the status of the queue of the other nodes. For this network we provided the stability region for the two and the three-user case. For the two-user case we provided the convexity conditions of the stability region. Furthermore, we provided a detailed mathematical analysis and derived the generating function of the stationary joint queue length distribution of user nodes in terms of the solution of a two boundary value problems. Based on that result we obtained expressions for the average queueing delay at each user node. For the two-user symmetric case with MPR we obtained a lower and an upper bound for the average delay without explicitly computing the generating function for the stationary joint queue length distribution. The bounds as shown in the numerical results appear to be tight. Explicit expressions for the average delay are obtained for the model with capture effect. Finally, we provided the optimal transmission probability in closed form expression that minimizes the average delay in the symmetric capture case.

\section*{Appendix}
\textbf{Intersection points of the curves:}
In the following we focus on the location of the intersection points of $R(x,y)=0$, $A(x,y)=0$ (resp. $B(x,y)$). These points (if exist) are potential singularities for the functions $H(x,0)$, $H(0,y)$, and thus, their investigation is crucial regarding the analytic continuation of  $H(x,0)$, $H(0,y)$ outside the unit disk. Note that the analytic continuation of $H(x,0)$ (resp. $H(0,y)$) outside the unit disc can be achieved by various methods (e.g., Lemma 2.2.1 and Theorem 3.2.3 in \cite{fay1}). 
\paragraph{Intersection points between $R(x,y)=0$, $A(x,y)=0$.}
Let $y\in \doubletilde{C}_{y}$ and $R(x,y) = 0$, $x = X_{\pm}(y)$. We can easily show that the resultant in $x$ of the two polynomials $R(x,y)$ and $A(x,y)$ is
\begin{displaymath}
\begin{array}{rl}
Res_{x}(R,A;y)=&y(y-1)Q(y),\\
Q(y)=&-\lambda_{2}d_{1}(d_{1}+(1+\lambda_{1})\alpha_{2}\widehat{\alpha}_{1})y^{2}+y\alpha_{2}\widehat{\alpha}_{1}[d_{1}(\lambda+\lambda_{1}\lambda_{2})\\&-\alpha_{1}^{*}\tilde{P}_{1/\{1\}}(d_{1}+\alpha_{2}\widehat{\alpha}_{1})]+(\alpha_{2}\widehat{\alpha}_{1})^{2}\alpha_{1}^{*}\tilde{P}_{1/\{1\}}.
\end{array}
\end{displaymath}

Note that $Q(0)=(\alpha_{2}\widehat{\alpha}_{1})^{2}\alpha_{1}^{*}\tilde{P}_{1/\{1\}}>0$ and $Q(1)=d_{1}[\lambda_{1}\alpha_{2}\widehat{\alpha}_{1}-\lambda_{2}d_{1}-\alpha_{2}\widehat{\alpha}_{1}\alpha_{1}^{*}\tilde{P}_{1/\{1\}}]>0$ due to the fact that $d_{1}<0$ and the stability condition (see Lemma \ref{lem}).

Similarly, for $x\in \doubletilde{C}_{x}$ and $R(x,y) = 0$, $y = Y_{\pm}(x)$. We can easily show that the resultant in $y$ of the two polynomials $R(x,y)$ and $A(x,y)$ is
\begin{displaymath}
\begin{array}{rl}
Res_{y}(R,A;x)=&x(x-1)\alpha_{2}\widehat{\alpha}_{1}Z(x),\\
Z(x)=&-\lambda_{1}(\alpha_{2}\widehat{\alpha}_{2}+(1+\lambda_{1})d_{1})x^{2}+x[(\lambda+\lambda_{1}\lambda_{2})d_{1}+(\alpha_{2}\widehat{\alpha}_{1}+d_{1})\alpha_{1}^{*}\tilde{P}_{1/\{1\}}]-\alpha_{1}^{*}\tilde{P}_{1/\{1\}}d_{1}.
\end{array}
\end{displaymath}
Note also that $Z(0)=-\alpha_{1}^{*}\tilde{P}_{1/\{1\}}d_{1}>0$ since $d_{1}<0$ and $Z(1)=\alpha_{2}\widehat{\alpha}_{1}\alpha_{1}^{*}\tilde{P}_{1/\{1\}}-\lambda_{1}\alpha_{2}\widehat{\alpha}_{1}+\lambda_{2}d_{1}>0$ due to the stability conditions (see Lemma \ref{lem}). If $\alpha_{1}^{*}\leq min\{1,\frac{\alpha_{2}\widehat{\alpha}_{1}+(1+\lambda_{2})\alpha_{1}\widehat{\alpha}_{2}}{(1+\lambda_{2})\tilde{P}_{1/\{1\}}}\}$, then $\lim_{x\to\infty}Z(x)=-\infty$, and $Z(x)=0$ has two roots of opposite sign, say $x_{*}<0<1<x^{*}$. If $\frac{\alpha_{2}\widehat{\alpha}_{1}+(1+\lambda_{2})\alpha_{1}\widehat{\alpha}_{2}}{(1+\lambda_{2})\tilde{P}_{1/\{1\}}}<\alpha_{1}^{*}\leq 1$, then $\lim_{x\to\infty}Z(x)=+\infty$, and $Z(x)=0$ has two positive roots , say $1<\tilde{x}_{*}<x_{3}<x_{4}<\tilde{x}^{*}$ (due to the stability conditions). In the former case we have to check if $x^{*}$ is in $S_{x}=G_{\mathcal{M}}\cap \{x:|x|>1\}$, while in the latter case if $\tilde{x}_{*}$ is in $S_{x}$. These zeros, if they lie in $S_{x}$ such that $|Y_{0}(x)|\leq1$, are poles of $H(x,y)$. Denote from hereon 
\begin{displaymath}
\bar{x}=\left\{\begin{array}{rl}
x^{*},&\alpha_{1}^{*}\leq min\{1,\frac{\alpha_{2}\widehat{\alpha}_{1}+(1+\lambda_{2})\alpha_{1}\widehat{\alpha}_{2}}{(1+\lambda_{2})\tilde{P}_{1/\{1\}}}\},\\
\tilde{x}_{*},&\frac{\alpha_{2}\widehat{\alpha}_{1}+(1+\lambda_{2})\alpha_{1}\widehat{\alpha}_{2}}{(1+\lambda_{2})\tilde{P}_{1/\{1\}}}<\alpha_{1}^{*}\leq 1.
\end{array}\right.
\end{displaymath}
\paragraph{Intersection points between $R(x,y)=0$, $B(x,y)=0$.} Let $y\in \doubletilde{C}_{y}$ and $R(x,y) = 0$, $x = X_{\pm}(y)$. It is easily shown that
\begin{displaymath}
\begin{array}{l}
R(x,y)=B(x,y)+\lambda_{1}(1-x)+\lambda_{2}(1-y)+\lambda_{1}\widehat{\lambda}_{2}(1-x)(1-y)+\alpha_{2}^{*}\tilde{P}_{2/\{2\}}(1-\frac{1}{y}). 
\end{array}
\end{displaymath}
Thus, $R(x,y)=0$, $B(x,y)=0$ implies that,
\begin{displaymath}
\begin{array}{rl}
\lambda_{1}xy(1-x)+\lambda_{2}xy(1-y)+\lambda_{1}\lambda_{2}xy(1-x)(1-y)+\alpha_{2}^{*}\tilde{P}_{2/\{2\}}x(y-1)=&0,\\
\alpha_{1}\widehat{\alpha}_{2}y(x-1)+d_{2}x(y-1)=&0.
\end{array}
\end{displaymath}
The second equation gives $x=\alpha_{1}\widehat{\alpha}_{2}y/(\alpha_{1}\widehat{\alpha}_{2}y+d_{2}(y-1))$. Substituting back in the first one yields,
$W(y)=\frac{y-1}{y(\alpha_{1}\widehat{\alpha}_{2}y+d_{2}(y-1))}S(y)$, where $S(y)=-\lambda_{2}(\alpha_{1}\widehat{\alpha}_{2}+d_{2}(1+\lambda_{1}))y^{2}+y(d_{2}(\lambda+\lambda_{1}\lambda_{1})+\alpha_{2}^{*}\tilde{P}_{2/\{2\}}(d_{2}+\alpha_{1}\widehat{\alpha}_{2}))-\alpha_{2}^{*}\tilde{P}_{2/\{2\}}d_{2}.$ Note that $S(0)=-\alpha_{2}^{*}\tilde{P}_{2/\{2\}}d_{2}>0$, $S(1)=\alpha_{1}\widehat{\alpha}_{2}\alpha_{2}^{*}\tilde{P}_{2/\{2\}}-\widehat{\lambda}_{2}\alpha_{1}\widehat{\alpha}_{2}+\widehat{\lambda}_{1}d_{2}>0$, due to the stability conditions. If $\alpha_{2}^{*}<min\left\{1,\frac{\alpha_{1}\widehat{\alpha}_{2}+(1+\lambda_{1})\alpha_{2}\widehat{\alpha}_{1}}{(1+\lambda_{1})\tilde{P}_{2/\{2\}}}\right\}$, then $\lim_{y\to \infty}S(y)=-\infty$, and $S(y)$ has two roots of opposite sign, say $y_{*}$, $y^{*}$ such that $y_{*}<0<1<y^{*}$, and $W(y)>0$ for $y\in(0,1)$, which in turn implies that $B(X_{0}(y),y)\neq0$, $y\in[y_{1},y_{2}]\subset(0,1)$, or equivalently $B(x,Y_{0}(x))\neq0$, $x\in\mathcal{M}$. In case $\frac{\alpha_{1}\widehat{\alpha}_{2}+(1+\lambda_{1})\alpha_{2}\widehat{\alpha}_{1}}{(1+\lambda_{1})\tilde{P}_{2/\{2\}}}<\alpha_{2}^{*}\leq1$, $\lim_{y\to \infty}S(y)=+\infty$, and $S(y)$ has two positive roots, say $\widehat{y}_{*}$, $\widehat{y}^{*}$ such that $1<\widehat{y}_{*}<y_{3}<y_{4}<\widehat{y}^{*}$, and $W(y)>0$ for $y\in(0,1)$, which in turn implies that $B(X_{0}(y),y)\neq0$, $y\in[y_{1},y_{2}]\subset(0,1)$, or equivalently $B(x,Y_{0}(x))\neq0$, $x\in\mathcal{M}$.

\bibliographystyle{ieeetr}
\bibliography{bibliography}
\end{document}